\documentclass[a4paper,USenglish,cleveref,autoref,thm-restate]{lipics-v2021}

\usepackage{ifthen}
\newboolean{llncs}
\setboolean{llncs}{false}
\ifthenelse{\boolean{llncs}}{
\usepackage[T1]{fontenc}
\usepackage{hyperref}
\usepackage{graphicx}
\usepackage{amsmath}
\usepackage{cleveref}
\usepackage{subcaption}
\usepackage{thm-restate}
}{}

\usepackage{xspace}
\usepackage{todonotes}
\usepackage{multirow}

\newcommand{\splitSP}{{\normalfont\texttt{Split\-Sep\-a\-ra\-tion}}\-{\normalfont\texttt{Pair}}\xspace}
\newcommand{\joinSP}{{\normalfont\texttt{Join\-Sep\-a\-ra\-tion}}\-{\normalfont\texttt{Pair}}\xspace}
\newcommand{\isolV}{{\normalfont\texttt{Iso\-late\-Ver\-tex}}\xspace}
\newcommand{\integV}{{\normalfont\texttt{In\-te\-grate}}\xspace}
\newcommand{\insertG}{{\normalfont\texttt{In\-sert\-Graph}}\xspace}
\newcommand{\insertGSPQR}{{\normalfont\ensuremath{\texttt{In\-sert}}\-\ensuremath{\texttt{Graph}_\texttt{SPQR}}}\xspace}
\newcommand{\mergeSPQR}{{\normalfont\ensuremath{\texttt{Merge}_\texttt{SPQR}}}\xspace}
\newcommand{\replace}[4]{\ensuremath{#1[#3\rightarrow_{#4}#2]}}

\newcommand{\skeldec}{skeleton decomposition\xspace}
\newcommand{\skeldecVars}[1][]{\ensuremath{\mathcal S#1=(\mathcal G#1,\allowbreak \origV#1,\allowbreak \origE#1,\allowbreak \twinE#1)}\xspace}
\newcommand{\eskeldec}{extended skeleton decomposition\xspace}
\newcommand{\eskeldecVars}[1][]{\ensuremath{\mathcal S#1=(\mathcal G#1,\allowbreak \origV#1,\allowbreak \origE#1,\allowbreak \twinE#1,\allowbreak \twinV#1)}\xspace}

\newcommand{\cplan}{{\normalfont\textsc{Clus}\-\textsc{tered} \textsc{Pla}\-\textsc{nari}\-\textsc{ty}}\xspace}
\newcommand{\pqplan}{{\normalfont\textsc{Syn}\-\textsc{chro}\-\textsc{nized} \textsc{Pla}\-\textsc{nari}\-\textsc{ty}}\xspace}

\newcommand{\contract}{{\normalfont\texttt{En}\-\texttt{cap}\-\texttt{su}\-\texttt{late}\-\texttt{And}\-\texttt{Join}}\xspace}
\newcommand{\propagate}{{\normalfont\texttt{Prop}\-\texttt{a}\-\texttt{gatePQ}}\xspace}
\newcommand{\simplify}{{\normalfont\texttt{Sim}\-\texttt{pli}\-\texttt{fy}\-\texttt{Match}\-\texttt{ing}}\xspace}

\DeclareMathOperator{\origV}{origV}
\DeclareMathOperator{\origE}{origE}
\DeclareMathOperator{\twinE}{twinE}
\DeclareMathOperator{\twinV}{twinV}

\newcommand{\cupdot}{\mathbin{\mathaccent\cdot\cup}}

\makeatletter
\NewDocumentCommand\namedlabel{omm}{%
  \begingroup%
  #3\def\@currentlabel{#3}%
  \phantomsection%
  \IfNoValueTF{#1}
  {\label{#2}}
  {\label[#1]{#2}}%
  \endgroup%
}
\NewDocumentCommand\nameditem{omm}{%
  \item[%
      \IfNoValueTF{#1}
      {\namedlabel{#2}{#3}}
      {\namedlabel[#1]{#2}{#3}}%
    ]%
}
\makeatother
\newcommand\cond[1]{\textsf{\ref{cond:#1}}}

\crefname{cond}{condition}{conditions}
\crefname{step}{step}{steps}

\usepackage{ifthen}
\newboolean{long}
\setboolean{long}{true}
\ifthenelse{\boolean{long}}{
  \newcommand{\proofomitted}{}
}{
  \newcommand{\proofomitted}{$\ast$}
  \let\origrestatable=\restatable
  \def\restatable{\origrestatable[\proofomitted]}
}

\graphicspath{{./graphics/}{./graphics/pdf/}}

\ifthenelse{\boolean{llncs}}{
\bibliographystyle{splncs04}

\title{Maintaining Triconnected Components under Node Expansion\thanks{Funded by DFG-grant RU-1903/3-1.}}

\author{
  Simon D. Fink\inst{1}\orcidID{0000-0002-2754-1195} \and
  Ignaz Rutter\inst{2}\orcidID{0000-0002-3794-4406}
}

\institute{
  Faculty of Informatics and Mathematics, University of Passau, Germany
  \email{\{finksim,rutter\}@fim.uni-passau.de}
}

\authorrunning{S.\,D. Fink and I. Rutter}

}{ 

\bibliographystyle{abbrvurl}

\title{Maintaining Triconnected Components under Node Expansion}

\author{Simon D. Fink}{Faculty of Informatics and Mathematics, University of Passau, Germany}{finksim@fim.uni-passau.de}{https://orcid.org/0000-0002-2754-1195}{}

\author{Ignaz Rutter}{Faculty of Informatics and Mathematics, University of Passau, Germany}{rutter@fim.uni-passau.de}{https://orcid.org/0000-0002-3794-4406}{}

\authorrunning{S.\,D. Fink and I. Rutter} 

\Copyright{Simon D. Fink and Ignaz Rutter}

\ccsdesc[500]{Mathematics of computing~Graph algorithms}

\keywords{SPQR-Tree, Dynamic Algorithm, Cluster Planarity} 

\funding{Funded by DFG-grant RU-1903/3-1.}

\nolinenumbers 

\hideLIPIcs  

} 

\begin{document}

\maketitle

\begin{abstract}
SPQR-trees are a central component of graph drawing and are also important in many further areas of computer science.
From their inception onwards, they have always had a strong relation to dynamic algorithms maintaining information, e.g., on planarity and triconnectivity, under edge insertion and, later on, also deletion.
In this paper, we focus on a special kind of dynamic update, the expansion of vertices into arbitrary biconnected graphs, while maintaining the SPQR-tree and further information.
This will also allow us to efficiently merge two SPQR-trees by identifying the edges incident to two vertices with each other.
We do this working along an axiomatic definition lifting the SPQR-tree to a stand-alone data structure that can be modified independently from the graph it might have been derived from.
Making changes to this structure, we can now observe how the graph represented by the SPQR-tree changes, instead of having to reason which updates to the SPQR-tree are necessary after a change to the represented graph.

Using efficient expansions and merges allows us to improve the runtime of the \pqplan algorithm by Bläsius et al.~\cite{Blaesius2021} from $O(m^2)$ to $O(m\cdot \Delta)$, where $\Delta$ is the maximum pipe degree.
This also reduces the time for solving several constrained planarity problems, e.g. for \cplan from $O((n+d)^2)$ to $O(n+d\cdot \Delta)$, where $d$ is the total number of crossings between cluster borders and edges and $\Delta$ is the maximum number of edge crossings on a single cluster border.%

\ifthenelse{\boolean{llncs}}{\keywords{SPQR-Tree \and Dynamic Algorithm \and Cluster Planarity}}{}
\end{abstract}

\newpage
\setcounter{page}{1}

\section{Introduction}
The SPQR-tree is a data structure that represents the decomposition of a graph at its \emph{separation pairs}, that is the pairs of vertices whose removal disconnects the graph.
The components obtained by this decomposition are called \emph{skeletons}.
SPQR-trees form a central component of many graph visualization techniques and are used for, e.g., planarity testing and variations thereof~\cite{Brueckner2019,Battista1996p,Gutwenger2002,Holm2020p,Poutre1994} and for computing embeddings and layouts~\cite{Angelini2009,Bienstock1990,Blaesius2016a,Didimo2020,Gutwenger2010,Weiskircher2002}; see~\cite{Mutzel2003} for a survey of graph drawing applications.
Outside of graph visualization they are used in the context of, e.g., minimum spanning trees~\cite{Bienstock1989,Battista1990}, triangulations~\cite{Biedl1997}, and crossing optimization~\cite{Gutwenger2010,Weiskircher2002}.
They also have multiple applications outside of graph theory and even computer science, e.g. for creating integrated circuits~\cite{Chen1999,Zhang2013}, business processes modelling~\cite{Vanhatalo2009}, electrical engineering~\cite{Franken2005}, theoretical physics~\cite{Manteuffel2012} and genomics~\cite{Fedarko2017}.

Initially, SPQR-trees were devised by Di Battista and Tamassia for incremental planarity testing~\cite{Battista1989,Battista1996p}.
As such, even in their initial form, SPQR-trees already allowed dynamic updates in the form of edge addition.
Their use was quickly expanded to other on-line problems~\cite{Battista1996,Battista1990}.
In addition to the applications mentioned above, this also sparked a series of further papers improving the runtime of the incremental data structure~\cite{Poutre1992,Poutre1994,Westbrook1992} and also extending it to be fully-dynamic, i.e., allowing insertion and deletion of vertices and edges, in $O(\sqrt{n})$ time~\cite{Eppstein1996,Galil1999}, where $n$ is the number of vertices in the graph.
Recently, Holm and Rotenberg described a fully-dynamic algorithm for maintaining planarity and triconnectivity information in $O(\log^3 n)$ time per operation \cite{Holm2020p,Holm2020} (see also there for a short history on dynamic SPQR-tree algorithms).

In this paper, we consider an incremental setting where we allow a single operation that expands a vertex $v$ into an arbitrary biconnected graph $G_\nu$.
Using the approach of Holm and Rotenberg~\cite{Holm2020p}, this takes $O((\deg(v)+|G_\nu|) \cdot \log^3 n)$ time by first removing $v$ and its incident edges and then incrementally inserting $G_\nu$.
We improve this to $O(\deg(v)+|G_\nu|)$ using an algorithm that is much simpler and thus also more likely to improve performance in practice.
In addition, our approach also allows to efficiently merge two SPQR-trees as follows.
Given two biconnected graphs $G_1, G_2$ containing vertices $v_1,v_2$, respectively, together with a bijection between their incident edges, we construct a new graph $G$ by replacing $v_1$ with $G_2-v_2$ in $G_1$, identifying edges using the given bijection.
Given the SPQR-trees of $G_1$ and $G_2$, we show that the SPQR-tree of $G$ can be found in $O(\deg(v_1))$ time.
More specifically, we present a data structure that supports the following operations:
$\insertGSPQR$ expands a single vertex in time linear in the size of the expanded subgraph,
$\mergeSPQR$ merges two SPQR-trees in time linear in the degree of the replaced vertices,
$\texttt{IsPlanar}$ indicates whether the currently represented graph is planar in constant time, and
$\texttt{Rotation}$ yields one of the two possible planar rotations of a vertex in a triconnected skeleton in constant time.
Furthermore, our data structure can be adapted to yield consistent planar embeddings for all triconnected skeletons and to test for the existence of three distinct paths between two arbitrary vertices with an additional factor of $\alpha(n)$ for all operations, where $\alpha$ is the inverse Ackermann function.

\begin{table}[t]
  \centering
  \renewcommand{\arraystretch}{1.4}
  \setlength{\tabcolsep}{4pt}
  \begin{tabular}{p{4cm}|c|c|c}
    \multirow{2}{*}{Problem}                  & \multicolumn{3}{c}{Running Times}                                                     \\
                                              & before \cite{Blaesius2021}              & using \cite{Blaesius2021} & with this paper \\ \hline
    Atomic Embeddability / \mbox{\pqplan}     & $O(m^8)$ \cite{Fulek2019}               & $O(m^2)$                  & $O(m\cdot\Delta)$   \\ \hline
    ClusterPlanarity                          & $O((n+d)^8)$ \cite{Fulek2019}           & $O((n+d)^2)$              & $O(n+d\cdot\Delta)$ \\ \hline
    Connected SEFE                            & $O(n^{16})$ \cite{Fulek2019}            & $O(n^2)$                  & $O(n\cdot\Delta)$   \\
                                              & bicon: $O(n^2)$ \cite{Blaesius2011}     &                           &                 \\ \hline
    Partially PQ-Constrained Planarity        & bicon: $O(m)$ \cite{Blaesius2011}       & $O(m^2)$                  & $O(m\cdot\Delta)$   \\ \hline
    Row-Column Independent NodeTrix Planarity & bicon: $O(n^2)$ \cite{Liotta2020}       & $O(n^2)$                  & $O(n\cdot\Delta)$   \\ \hline
    Strip Planarity                           & $O(n^8)$ \cite{Angelini2016a,Fulek2019} & $O(n^2)$                  & $O(n\cdot\Delta)$   \\
                                              & fixed emb: poly \cite{Angelini2016a}    &                           &
  \end{tabular}
  \caption{
    The best known running times for various constrained planarity problems before \pqplan~\cite{Blaesius2021} was published;
    using it as described in~\cite{Blaesius2021}; and using it together with the speed-up from this paper.
    Running times prefixed with ``bicon'' only apply for certain problem instances which expose some form of biconnectivity.
    The variables $n$ and $m$ refer to the number of vertices and edges of the problem instance, respectively.
    The variable $d$ refers to the number of edge-cluster boundary crossings in \cplan instances,
    while $\Delta$ refers to the maximum pipe degree in the corresponding \pqplan instances.
    This is bounded by the maximum number of edges crossing a single cluster border or the maximum vertex degree in the input instance, depending on the problem.
  }
  \label{tab:constplanprobs}
\end{table}

The main idea of our approach is that the subtree of the SPQR-tree affected by expanding a vertex $v$ has size linear in the degree of $v$, but may contain arbitrarily large skeletons.
In a ``non-normalized'' version of an SPQR-tree, the affected cycle (`S') skeletons can easily be split to have a constant size, while we develop a custom splitting operation to limit the size of triconnected `R' skeletons.
This limits the size of the affected structure to be linear in the degree of $v$ and allows us to perform the expansion efficiently.

In addition to the description of this data structure, the technical contribution of this paper is twofold:
First, we develop an axiomatic definition of the decomposition at separation pairs, putting the SPQR-tree as ``mechanical'' data structure into focus instead of relying on and working along a given graph structure.
As a result, we can deduce the represented graph from the data structure instead of computing the data structure from the graph.
This allows us to make more or less arbitrary changes to the data structure (respecting its consistency criteria) and observe how the graph changes, instead of having to reason which changes to the graph require which updates to the data structure.

Second, we explain how our data structure can be used to improve the runtime of the algorithm by Bläsius et al.~\cite{Blaesius2021} for solving \pqplan from $O(m^2)$ to $O(m\cdot \Delta)$, where $\Delta$ is the maximum pipe degree (i.e. the maximum degree of a vertex with synchronization constraints that enforce its rotation to be the same as that of another vertex).
\pqplan can be used to model and solve a vast class of different kinds of constrained planarity, see \Cref{tab:constplanprobs} for an overview of problems benefiting from this speedup.
Among them is the notorious \cplan, whose complexity was open for 30 years before Fulek and Tóth gave an algorithm with runtime $O((n+d)^8)$ in 2019~\cite{Fulek2019}, where $d$ is the total number of crossings between cluster borders and edges.
Shortly thereafter, Bläsius et al.~\cite{Blaesius2021} gave a solution in $O((n+d)^2)$ time.
We improve this to $O(n+d\cdot \Delta)$, where $\Delta$ is the maximum number of edge crossings on a single cluster border.

This work is structured as follows.
\Cref{sec:prelim} contains an overview of the definitions used in this work.
In \Cref{sec:skeldec}, we describe the \skeldec and show how it relates to the SPQR-tree.
\Cref{sec:eskeldec} extends this data structure by the capability of splitting triconnected components.
In \Cref{sec:dyn-eskeldec}, we exploit this feature to ensure the affected part of the SPQR-tree is small when we replace a vertex with a new graph.
\ifthenelse{\boolean{long}}{%
\Cref{sec:application} contains more details on the background of \textsc{Syn}\-\textsc{chro}\-\textsc{nized} and \cplan and shows how our results can be used to reduce the time required for solving them.
}{%
\Cref{sec:application} shows how this can be used to reduce the time required for solving \textsc{Syn}\-\textsc{chro}\-\textsc{nized} and \cplan.
Due to space constraints, proofs of statements marked with a star are only given in the appendix.
\Cref{sec:background} contains more details on the background of \textsc{Syn}\-\textsc{chro}\-\textsc{nized} and \cplan together with a discussion of our results.
}

\section{Preliminaries}\label{sec:prelim}

In the context of this work, $G=(V,E)$ is a (usually biconnected and loop-free) multi-graph with $n$ vertices $V$ and $m$ (possibly parallel) edges $E$.
For a vertex $v$, we denote its open neighborhood (excluding $v$ itself) by $N(v)$.
For a bijection or matching $\phi$ we call $\phi(x)$ the \emph{partner} of an element $x$.
We use $A \cupdot B$ to denote the union of two disjoint sets $A, B$.

A separating $k$-set is a set of $k$ vertices whose removal increases the number of connected components.
Separating 1-sets are called \emph{cutvertices}, while separating 2-sets are called \emph{separation pairs}.
A connected graph is \emph{biconnected} if it does not have a cutvertex.
A biconnected graph is \emph{triconnected} if it does not have a separation pair.
Maximal biconnected subgraphs are called \emph{blocks}.
Each separation pair divides the graph into \emph{bridges}, the maximal subgraphs which cannot be disconnected by removing or splitting the vertices of the separation pair.
%
A \emph{bond} is a graph that consists solely of two \emph{pole} vertices connected by multiple parallel edges, a \emph{polygon} is a simple cycle, while a \emph{rigid} is any simple triconnected graph.
A \emph{wheel} is a cycle with an additional central vertex connected to all other vertices.

Finally, the \emph{expansion} that is central to this work is formally defined as follows.
Let $G_\alpha, G_\beta$ be two graphs where $G_\alpha$ contains a vertex $u$ and $G_\beta$ contains $|N(u)|$ marked vertices, together with a bijection $\phi$ between the neighbors of $u$ and the marked vertices in $G_\beta$.
With \replace{G_\alpha}{G_\beta}{u}{\phi} we denote the graph that is obtained from the disjoint union of $G_\alpha,G_\beta$ by identifying each neighbor $x$ of $u$ with its respective marked vertex $\phi(x)$ in $G_\beta$ and removing $u$, i.e. the graph $G_\alpha$ where the vertex $u$ was expanded into $G_\beta$.%

\section{Skeleton Decompositions}\label{sec:skeldec}
A \emph{skeleton structure} \skeldecVars that \emph{represents} a graph $G_{\mathcal S}=(V,E)$ consists of a set $\mathcal G$ of disjoint \emph{skeleton} graphs together with three total, surjective mappings $\twinE, \origE$, and $\origV$ that satisfy the following conditions:
\begin{itemize}
\item Each skeleton $G_\mu=(V_\mu, E_\mu^\mathrm{real} \cupdot E_\mu^\mathrm{virt})$ in $\mathcal G$ is a multi-graph where each edge is either in $E_\mu^\mathrm{real}$ and thus called \emph{real} or in $E_\mu^\mathrm{virt}$ and thus called \emph{virtual}.
\item Bijection $\twinE : E^\mathrm{virt} \rightarrow E^\mathrm{virt}$ matches all virtual edges $E^\mathrm{virt} = \bigcup_\mu E_\mu^\mathrm{virt}$ such that $\twinE(e)\neq e$ and $\twinE^2=\text{id}$.
\item Surjection $\origV : \bigcup_\mu V_\mu \rightarrow V$ maps all skeleton vertices to graph vertices.
\item Bijection $\origE : \bigcup_\mu E_\mu^\mathrm{real} \rightarrow E$ maps all real edges to the graph edge set $E$.
\end{itemize}
Note that each vertex and each edge of each skeleton is in the domain of exactly one of the three mappings.
As the mappings are surjective, $V$ and $E$ are exactly the images of $\origV$ and $\origE$.
For each vertex $v\in G_{\mathcal S}$, the skeletons that contain an \emph{allocation vertex} $v'$ with $\origV(v')=v$ are called the \emph{allocation skeletons} of $v$.
Furthermore, let $T_{\mathcal S}$ be the graph where each node $\mu$ corresponds to a skeleton $G_\mu$ of $\mathcal G$.
Two nodes of $T_{\mathcal S}$ are adjacent if their skeletons contain a pair of virtual edges matched with each other.

We call a skeleton structure a \emph{\skeldec} if it satisfies the following conditions:
\begin{description}
\nameditem[cond]{cond:bicon}{1~(bicon)}
  Each skeleton is biconnected.
\nameditem[cond]{cond:tree}{2~(tree)}
  Graph $T_{\mathcal S}$ is simple, loop-free, connected and acyclic, i.e., a tree.
\nameditem[cond]{cond:orig-inj}{3~(orig-inj)}
  For each skeleton $G_\mu$, the restriction $\origV | _{V_\mu}$ is injective.
\nameditem[cond]{cond:orig-real}{4~(orig-real)}
  For each real edge $uv$, the endpoints of $\origE(uv)$ are $\origV(u)$ and $\origV(v)$.
\nameditem[cond]{cond:orig-virt}{5~(orig-virt)}
  Let $uv$ and $u'v'$ be two virtual edges with $uv=\twinE(u'v')$.
  For their respective skeletons $G_\mu$ and $G_\mu'$ (where $\mu$ and $\mu'$ are adjacent in $T_{\mathcal S}$), it is
  $\origV(V_\mu) \cap \origV(V_{\mu'}) = \origV(\{u,v\})=\origV(\{u',v'\})$.
\nameditem[cond]{cond:subgraph}{6~(subgraph)}
  The allocation skeletons of any vertex of $G_{\mathcal S}$ form a connected subgraph of~$T_{\mathcal S}$.
\end{description}
\Cref{fig:skeldec} shows an example of $\mathcal S$, $G_{\mathcal S}$, and $T_{\mathcal S}$.
We call a \skeldec with only one skeleton $G_\mu$ \emph{trivial}.
Note that in this case, $G_\mu$ is isomorphic to $G_{\mathcal S}$, and $\origE$ and $\origV$ are actually bijections between the edges and vertices of both graphs.

\begin{figure}[t]
  \centering
  \begin{subfigure}[t]{\linewidth}
    \includegraphics[page=7,width=\linewidth]{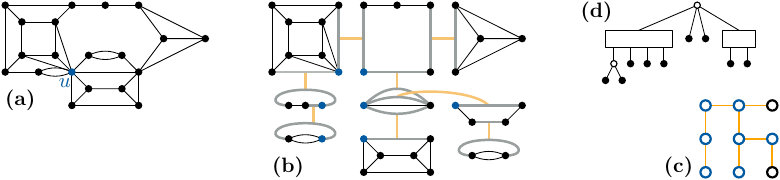}
    \phantomsubcaption\label{fig:skeldec-graph}
    \phantomsubcaption\label{fig:skeldec-skeletons}
    \phantomsubcaption\label{fig:skeldec-tree}
    \phantomsubcaption\label{fig:embedding-tree}
  \end{subfigure}


  \caption{
    Different views on the \skeldec $\mathcal S$.
    \textbf{(a)} The graph $G_{\mathcal S}$ with a vertex $u$ marked in blue.
    \textbf{(b)} The skeletons of $\mathcal G$.
      Virtual edges are drawn in gray with their matching $\twinE$ being shown in orange.
      The allocation vertices of $u$ are marked in blue.
    \textbf{(c)} The tree $T_{\mathcal S}$. 
      The allocation skeletons of $u$ are marked in blue.
    \textbf{(d)} The embedding tree of vertex $u$ as described in \Cref{sec:embed-trees}.
      P-nodes are shown as white disks, Q-nodes are shown as large rectangles.
      The leaves of the embedding tree correspond to the edges incident to $u$.
  }
  \label{fig:skeldec}
\end{figure}

To model the decomposition into triconnected components, we define the operations \splitSP and its converse, \joinSP, on a \skeldec \skeldecVars.
  For $\splitSP$, let $u,v$ be a separation pair of skeleton $G_\mu$ and let $(A,B)$ be a non-trivial bipartition of the bridges between $u$ and $v$.%
  \footnote{Note that a bridge might consist out of a single edge between $u$ and $v$ and that each bridge includes the vertices $u$ and $v$.}
  Applying $\splitSP(\mathcal S, (u, v)$, $(A, B))$ yields a \skeldec \skeldecVars['] as follows.
  In $\mathcal G'$, we replace $G_\mu$ by two skeletons $G_\alpha, G_\beta$, where $G_\alpha$ is obtained from $G_\mu[A]$ by adding a new virtual edge $e_\alpha$ between $u$ and $v$.
  The same respectively applies to $G_\beta$ with $G_\mu[B]$ and $e_\beta$.
  We set $\twinE'(e_\alpha)=e_\beta$ and $\twinE'(e_\beta)=e_\alpha$.
  Note that $\origV$ maps the endpoints of $e_\alpha$ and $e_\beta$ to the same vertices.
  All other skeletons and the mappings defined on them remain unchanged.

  For $\joinSP$, consider virtual edges $e_\alpha,e_\beta$ with $\twinE(e_\alpha)=e_\beta$ and let $G_\beta\neq G_\alpha$ be their respective skeletons.
  Applying $\joinSP(\mathcal S, e_\alpha)$ yields a \skeldec \skeldecVars['] as follows.
  In $\mathcal G'$, we merge $G_\alpha$ with $G_\beta$ to form a new skeleton $G_\mu$ by identifying the endpoints of $e_\alpha$ and $e_\beta$ that map to the same vertex of $G_\mathcal{S}$.
  Additionally, we remove $e_\alpha$ and $e_\beta$.
  All other skeletons and the mappings defined on them remain unchanged.

The main feature of both operations is that they leave the graph represented by the \skeldec unaffected while splitting a node or contracting and edge in $T_{\mathcal S}$, which can be verified by checking the individual conditions.

\begin{restatable}{lemma}{lemSplitJoin}
Applying \splitSP or \joinSP on a \skeldec \skeldecVars yields a \skeldec $\mathcal S'=(\mathcal G'$, $\origV'$, $\origE'$, $\twinE')$ with an unchanged represented graph $G_\mathcal{S'}=G_{\mathcal S}$.
\end{restatable}
\newcommand{\proofSplitJoin}{
\begin{proof}
We first check that all conditions still hold in the \skeldec $\mathcal S'$ returned by \splitSP.
As $(A,B)$ is a non-trivial bipartition, each set contains at least one bridge.
Together with $e_\alpha$ (and $e_\beta$), this bridge ensures that $G_\alpha$ (and $G_\beta$) remain biconnected, satisfying condition \cond{bicon}.
The operation splits a node $\mu$ of $T_\mathcal{S}$ into two adjacent nodes $\alpha, \beta$, whose neighbors are defined exactly by the virtual edges in $A, B$, respectively.
Thus, condition \cond{tree} remains satisfied.
The mappings $\origV', \origE'$ and $\twinE'$ obviously still satisfy conditions \cond{orig-inj} and \cond{orig-real}.
We duplicated exactly two nodes, $u$ and $v$ of adjacent skeletons $G_\alpha$ and $G_\beta$.
Because \cond{orig-inj} holds for $G_\mu$, $G_\alpha$ and $G_\beta$ share no other vertices that map to the same vertex of $G_\mathcal{S'}$.
Thus, condition \cond{orig-virt} remains satisfied.

Condition \cond{subgraph} could only be violated if the subgraph of $T_\mathcal{S'}$ formed by the allocation skeletons of some vertex $z \in G_\mathcal{S'}$ was no longer connected.
This could only happen if only one of $G_\alpha$ and $G_\beta$ were an allocation skeleton of $z$, while the other has a further neighbor that is also an allocation skeleton of $z$.
Assume without loss of generality that $G_\alpha$ and the neighbor $G_\nu$ of $G_\beta$, but not $G_\beta$ itself, were allocation skeletons of $z$.
Because $G_\nu$ and $G_\beta$ are adjacent in $T_\mathcal{S'}$ there are virtual edges $xy=\twinE'(x'y')$ with $xy\in G_\beta$ and $x'y'\in G_\nu$.
The same virtual edges are also present in the input instance, only with the difference that $xy\in G_\mu$ and $\mu$ (instead of $\beta$) and $\nu$ are adjacent in $T_\mathcal{S}$.
As the input instance satisfies condition \cond{orig-virt}, it is $z \in \origV(V_\nu) \cap \origV(V_\mu) = \origV(\{x,y\})=\origV(\{x',y'\})$.
As $\origV(\{x,y\})=\origV'(\{x,y\})$, this is a contradiction to $G_\beta$ not being an allocation skeleton of $z$.

Finally, the mapping $\origE$ remains unchanged and the only change to $\origV$ is to include two new vertices mapping to already existing vertices.
Due to condition \cond{orig-real} holding for both the input and the output instance, this cannot affect the represented graph $G_\mathcal{S'}$.

Now consider the \skeldec $\mathcal S'$ returned by \joinSP.
Identifying distinct vertices of distinct connected components does not affect their biconnectivity, thus condition \cond{bicon} remains satisfied.
The operation effectively contracts and removes an edge in $T_\mathcal{S}$, which does not affect $T_\mathcal{S'}$ being a tree satisfying condition \cond{tree}.
Note that condition \cond{tree} holding for the input instance also ensures that $G_\alpha$ and $G_\beta$ are two distinct skeletons.
As the input instance also satisfies condition \cond{orig-virt}, there are exactly two vertices in each of the two adjacent skeletons $G_\alpha$ and $G_\beta$, where $\origV$ maps to the same vertex of $G_\mathcal{S}$.
These two vertices must be part of the $\twinE$ pair making the two skeletons adjacent, thus they are exactly the two pairs of vertices we identify with each other.
Thus, $\origV|_{V_\mu}$ is still injective, satisfying condition \cond{orig-inj}.
As we modify no real edges and no other virtual edges, the mappings $\origV'$ and $\origE'$ obviously still satisfy condition \cond{orig-real}.
As the allocation skeletons of each graph vertex form a connected subgraph, joining two skeletons cannot change the intersection with any of their neighbors, leaving \cond{orig-virt} satisfied.
Finally, contracting a tree edge cannot lead to any of the subgraphs of \cond{subgraph} becoming disconnected, thus the condition also remains satisfied.
Again, no changes were made to $\origE$, while condition \cond{orig-virt} makes sure that $\origV$ mapped the two pairs of merged vertices to the same vertex of $G_{\mathcal S}$.
Thus, the represented graph $G_\mathcal{S'}$ remains unchanged.
\end{proof}
}
\ifthenelse{\boolean{long}}{\proofSplitJoin}{}

\newcommand{\ExhaustiveJoin}{
\begin{lemma}\label{lem:ExhaustiveJoin}
Exhaustively applying \joinSP to a \skeldec \skeldecVars yields a trivial \skeldec \skeldecVars['] where $\origE'$ and $\origV'$ define an isomorphism between $G_\mu'$ and $G_\mathcal{S'}$.
\end{lemma}
\begin{proof}
As all virtual edges are matched, and the matched virtual edge always belongs to a different skeleton (condition \cond{tree} ensures that $T_\mathcal{S}$ is loop-free), we can always apply \joinSP on a virtual edge until there are none left.
As $T_\mathcal{S}$ is connected, this means that the we always obtain a tree with a single node, that is an instance with only a single skeleton.
As a single application of \joinSP preserves the represented graph, any chain of multiple applications also does.
Note that $\origE'$ is a bijection and the surjective $\origV'$ is also injective on the single remaining skeleton due to condition \cond{orig-inj}, thus it also globally is a bijection.
Together with condition \cond{orig-real}, this ensures that any two vertices $u$ and $v$ of $G_\mu'$ are adjacent if and only if $\origV'(u)$ and $\origV'(v)$ are adjacent in $G_\mathcal{S'}$.
Thus $\origV'$ is an edge-preserving bijection, that is an isomorphism.
\end{proof}
}

\ifthenelse{\boolean{long}}{%
This gives us a second way of finding the represented graph by exhaustively joining all skeletons until there is only one left, obtaining the unique trivial skeleton decomposition:
\ExhaustiveJoin

}{%
Note that this gives us a second way of finding the represented graph by exhaustively joining all skeletons until there is only one left, the unique trivial skeleton decomposition (see also \Cref{lem:ExhaustiveJoin} in the Appendix).
}%
A key point about the \skeldec and especially the operation \splitSP now is that they model the decomposition of a graph at separation pairs.
This decomposition was formalized as \emph{SPQR-tree} by Di Battista and Tamassia~\cite{Battista1989} and is unique for a given graph~\cite{Hopcroft1973,MacLane1937}; see also~\cite{Gutwenger2010,Gutwenger2001}.
Angelini et al.~\cite{Angelini2014} describe a decomposition tree that is conceptually equivalent to our \skeldec.
They also present an alternative definition for the SPQR-tree as a decomposition tree satisfying further properties.
We adopt this definition for our \skeldec{}s as follows, not requiring planarity of triconnected components and allowing virtual edges and real edges to appear within one skeleton (i.e., having leaf Q-nodes merged into their parents).

\begin{definition}\label{thm:spqrt}\label{def:spqrt}
A \skeldec \skeldecVars where any skeleton in $\mathcal G$ is either a polygon, a bond, or triconnected (``rigid''),
and two skeletons adjacent in $T_\mathcal{S}$ are never both polygons or both bonds,
is the unique SPQR-tree of $G_\mathcal{S}$.
\end{definition}

The main difference between the well-known ideas behind decomposition trees and our \skeldec is that the latter allow an axiomatic access to the decomposition at separation pairs.
For the \skeldec, we employ a purely functional, ``mechanical'' data structure instead of relying on and working along a given graph structure.
In our case, the represented graph is deduced from the data structure (i.e. SPQR-tree) instead of computing the data structure from the graph.

\section{Extended Skeleton Decompositions}\label{sec:eskeldec}

Note that most skeletons, especially polygons and bonds, can easily be decomposed into smaller parts.
The only exception to this are triconnected skeletons which cannot be split further using the operations we defined up to now.
This is a problem when modifying a vertex that occurs in triconnected skeletons that may be much bigger than the direct neighborhood of the vertex.
To fix this, we define a further set of operations which allow us to isolate vertices out of arbitrary triconnected components by replacing them with a (``virtual'') placeholder vertex.
This placeholder then points to a smaller component that contains the actual vertex, see \Cref{fig:isolV}.
Modification of the edges incident to the placeholder is disallowed, which is why we call them ``occupied''.

Formally, the structures needed to keep track of the components split in this way in an \emph{extended} \skeldec \eskeldecVars are defined as follows.
Skeletons now have the form $G_\mu=(V_\mu \cupdot V_\mu^\mathrm{virt}, \  E_\mu^\mathrm{real} \cupdot E_\mu^\mathrm{virt} \cupdot E_\mu^\mathrm{occ})$.
Bijection $\twinV : V^\mathrm{virt} \rightarrow V^\mathrm{virt}$ matches all \emph{virtual vertices} $V^\mathrm{virt} = \bigcup_\mu V_\mu^\mathrm{virt}$, such that $\twinV(v)\neq v$, $\twinV^2=\text{id}$.
The edges incident to virtual vertices are contained in $E_\mu^\mathrm{occ}$ and thus considered \emph{occupied}; see \Cref{fig:isolV-post}.
Similar to the virtual edges matched by $\twinE$, any two virtual vertices matched by $\twinV$ induce an edge between their skeletons in $T_\mathcal{S}$.
Condition \cond{tree} also equally applies to those edges induced by $\twinV$, which in particular ensures that there are no parallel $\twinE$ and $\twinV$ tree edges in $T_\mathcal{S}$.
Similarly, the connected subgraphs of condition \cond{subgraph} can also contain tree edges induced by $\twinV$.
All other conditions remain unchanged, but we add two further conditions to ensure that $\twinV$ is consistent:

\begin{description}
\nameditem[cond]{cond:stars}{7~(stars)}
  For each $v_\alpha,v_\beta$ with $\twinV(v_\alpha)=v_\beta$, it is $\deg(v_\alpha)=\deg(v_\beta)$.
  All edges incident to $v_\alpha$ and $v_\beta$ are occupied and have distinct endpoints (except for $v_\alpha$ and $v_\beta$).
  Conversely, each occupied edge is adjacent to exactly one virtual vertex.
\nameditem[cond]{cond:orig-stars}{8~(orig-stars)}
  Let $v_\alpha$ and $v_\beta$ again be two virtual vertices matched with each other by $\twinV$.
  For their respective skeletons $G_\alpha$ and $G_\beta$ (where $\alpha$ and $\beta$ are adjacent in $T_{\mathcal S}$), it is
  $\origV(V_\alpha) \cap \origV(V_{\beta}) = \origV(N(v_\alpha))=\origV(N(v_\beta))$.
\end{description}

Note that both conditions together yield a bijection $\gamma_{v_\alpha v_\beta}$ between the neighbors of $v_\alpha$ and $v_\beta$,
as $\origV$ is injective when restricted to a single skeleton (condition \cond{orig-inj}) and $\deg(v_\alpha)=\deg(v_\beta)$.
Operations \splitSP and \joinSP can also be applied to an \eskeldec, yielding an \eskeldec without modifying $\twinV$.
To ensure that conditions \cond{stars} and \cond{orig-stars} remain unaffected by both operations, \splitSP cannot be applied if a vertex of the separation pair is virtual.

\begin{figure}[t]
  \begin{subfigure}[t]{0.35\textwidth}
    \centering
    \includegraphics[page=1]{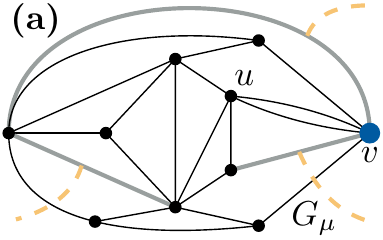}
    \phantomsubcaption
    \label{fig:isolV-pre}
  \end{subfigure}%
  \begin{subfigure}[t]{0.65\textwidth}
    \centering
    \includegraphics[page=2]{graphics/isolV}
    \phantomsubcaption
    \label{fig:isolV-post}
  \end{subfigure}

  \caption{
    \textbf{(a)} A triconnected skeleton $G_\mu$ with a highlighted vertex $v$ incident to two gray virtual edges.
    \textbf{(b)} The result of applying \isolV to isolate $v$ out of the skeleton.
      The red occupied edges in the old skeleton $G_\alpha$ form a star with center $v_\alpha$, while the red occupied edges in $G_\beta$ connect all neighbors of $v$ to form a star with center $v_\beta\neq v$.
      The centers $v_\alpha$ and $v_\beta$ are virtual and matched with each other.
      Neighbor $u$ of $v$ was split into vertices $u_\alpha$ and $u_\beta$.
  }
  \label{fig:isolV}
\end{figure}

The operations \isolV and \integV now allow us to isolate vertices out of triconnected components and integrate them back in, respectively.
  For $\isolV$, let $v$ be a non-virtual vertex of skeleton $G_\mu$, such that $v$ has no incident occupied edges.
  Applying $\isolV(\mathcal S, v)$ on an \eskeldec $\mathcal S$ yields an \eskeldec \eskeldecVars['] as follows.
  Each neighbor $u$ of $v$ is split into two non-adjacent vertices $u_\alpha$ and $u_\beta$,
  where $u_\beta$ is incident to all edges connecting $u$ with $v$, while $u_\alpha$ keeps all other edges of $u$.
  We set $\origV'(u_\alpha)=\origV'(u_\beta)=\origV(u)$.
  This creates an independent, star-shaped component with center $v$, which we move to skeleton $G_\beta$, while we rename skeleton $G_\mu$ to $G_\alpha$.
  We connect all $u_\alpha$ to a single new virtual vertex $v_\alpha\in V_{\alpha}^\textrm{virt}$ using occupied edges, and all $u_\beta$ to a single new virtual vertex $v_\beta\in V_{\beta}^\textrm{virt}$ using occupied edges; see \Cref{fig:isolV}.
  Finally, we set $\twinV'(v_\alpha)=v_\beta$, $\twinV'(v_\beta)=v_\alpha$, and add $G_\beta$ to $\mathcal G'$.
  All other mappings and skeletons remain unchanged.

  For $\integV$, consider two virtual vertices $v_\alpha,v_\beta$ with $\twinV(v_\alpha)=v_\beta$ and the bijection~$\gamma_{v_\alpha v_\beta}$ between the neighbors of $v_\alpha$ and $v_\beta$.
  An application of $\integV(\mathcal S, (v_\alpha, v_\beta))$ yields an \eskeldec \eskeldecVars['] as follows.
  We merge both skeletons into a skeleton $G_\mu$ (also replacing both in $\mathcal G'$) by identifying the neighbors of $v_\alpha$ and $v_\beta$ according to $\gamma_{v_\alpha v_\beta}$.
  Furthermore, we remove $v_\alpha$ and $v_\beta$ together with their incident occupied edges.
  All other mappings and skeletons remain unchanged.


\begin{restatable}{lemma}{lemIsolInteg}
Applying \isolV or \integV on an \eskeldec \eskeldecVars yields an \eskeldec \eskeldecVars['] with $G_\mathcal{S'}=G_{\mathcal S}$.
\end{restatable}
\newcommand{\proofIsolInteg}{
\begin{proof}
We first check that all conditions still hold in the \eskeldec $\mathcal S'$ returned by \isolV.
Condition \cond{bicon} remains satisfied, as the structure of $G_\alpha$ remains unchanged compared to $G_\mu$ and the skeleton $G_\beta$ is a bond.
As we are again splitting a node of $T_\mathcal{S}$, condition \cond{tree} also remains satisfied.
Due to the neighbors of $v_\beta$ and $v_\alpha$ mapping to the same vertices of $G_\mathcal{S'}$, conditions \cond{orig-inj}, \cond{orig-real}, and \cond{orig-virt} remain satisfied.
Conditions \cond{stars} and \cond{orig-stars} are satisfied by construction.

Lastly, condition \cond{subgraph} could only be violated if the subgraph of $T_\mathcal{S'}$ formed by the allocation skeletons of some vertex $z \in G_\mathcal{S'}$ was no longer connected.
This could only happen if only one of $G_\alpha$ and $G_\beta$ were an allocation skeleton of $z$, while the other has a further neighbor $G_\nu$ that is also an allocation skeleton of $z$.
Note that in any case, $\nu$ is adjacent to $\mu$ in $T_\mathcal{S}$ and $\mu$ must be an allocation skeleton of $z$, thus it is $z\in\origV(G_\nu)\cap\origV(G_\mu)$.
Depending on the adjacency of $\nu$, it is either $\origV(G_\nu)\cap\origV(G_\mu)=\origV'(G_\nu)\cap\origV(G_\alpha)$ or $\origV(G_\nu)\cap\origV(G_\mu)=\origV'(G_\nu)\cap\origV(G_\beta)$, as $\nu$ is not modified by the operation and both $\mathcal S$ and $\mathcal S'$ satisfy \cond{orig-virt} and \cond{orig-stars}.
This immediately contradicts the skeleton of $\{\alpha,\beta\}$, that is adjacent to $\nu$, not being an allocation skeleton of $z$.

Finally, the mapping $\origE$ remains unchanged and the only change to $\origV$ is to include some duplicated vertices mapping to already existing vertices.
Due to condition \cond{orig-real} holding for both the input and the output instance, this cannot affect the represented graph $G_\mathcal{S'}$.

Now consider the \eskeldec $\mathcal S'$ returned by \integV.
The merged skeleton is biconnected, as we are effectively replacing a single vertex by a connected subgraph, satisfying \cond{bicon}.
The operation effectively contracts and removes an edge in $T_\mathcal{S}$, which does not affect $T_\mathcal{S'}$ being a tree, satisfying condition \cond{tree}.
Note that condition \cond{tree} holding for the input instance also ensures that $v_\alpha$ and $v_\beta$ belong to two distinct skeletons.
As the input instance satisfies condition \cond{orig-virt}, the vertices in each of the two adjacent skeletons where $\origV$ maps to the same vertex of $G_\mathcal{S}$ are exactly the neighbors of the matched $v_\alpha$ and $v_\beta$.
Thus, $\origV|_{V_\alpha}$ is still injective, satisfying condition \cond{orig-inj}.
As we modify no real or virtual edges, the mappings $\origV', \origE'$ and $\twinE'$ obviously still satisfy conditions \cond{orig-real} and \cond{orig-virt}.
Finally, contracting a tree edge cannot lead to any of the subgraphs of \cond{subgraph} becoming disconnected, thus the condition also remains satisfied.
Conditions \cond{stars} and \cond{orig-stars} also remain unaffected, as we simply remove an entry from $\twinV$.

Again, no changes were made to $\origE$, while condition \cond{orig-stars} makes sure that $\origV$ mapped each pair of merged vertices to the same vertex of $G_{\mathcal S}$.
Thus, the represented graph $G_\mathcal{S'}$ remains unchanged.
\end{proof}
}
\ifthenelse{\boolean{long}}{\proofIsolInteg}{}

Furthermore, as \integV is the converse of \isolV and has no preconditions, any changes made by \isolV can be undone at any time to obtain a (non-extended) \skeldec, and thus possibly the SPQR-tree of the represented graph.

\begin{remark}\label{obs:ne-skeldec}
Exhaustively applying \integV to an \eskeldec \eskeldecVars yields a \eskeldec \eskeldecVars['] where $\twinV'=\emptyset$.
Thus, $\mathcal S'$ is equivalent to a (non-extended) \skeldec \skeldecVars['].
\end{remark}

\section{Node Expansion in Extended Skeleton Decompositions}\label{sec:dyn-eskeldec}

\begin{figure}[t]
  \begin{subfigure}[t]{.33\textwidth}
    \centering
    \includegraphics[page=8]{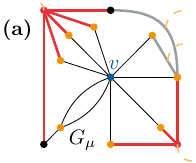}
    \phantomsubcaption
    \label{fig:insertGSPQR-pre}
  \end{subfigure}%
  \hfill
  \begin{subfigure}[t]{.33\textwidth}
    \centering
    \includegraphics[page=10]{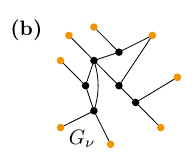}
    \phantomsubcaption
    \label{fig:insertGSPQR-graph}
  \end{subfigure}%
  \hfill
  \begin{subfigure}[t]{.33\textwidth}
    \centering
    \includegraphics[page=12]{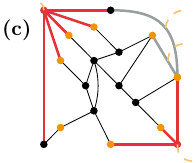}
    \phantomsubcaption
    \label{fig:insertGSPQR-integ}
  \end{subfigure}

  \caption{
    Expanding a skeleton vertex $v$ into a graph $G_\nu$ in the SPQR-tree of \Cref{fig:eskeldec-join-alloc}.
    \textbf{(a)}~The single allocation skeleton $G_\mu$ of $u$ with the single allocation vertex $v$ of $u$ from \Cref{fig:eskeldec-join-alloc}.
      The neighbors of $v$ are marked in orange.
    \textbf{(b)}~The inserted graph $G_\nu$ with orange marked vertices.
      Note that the graph is biconnected when all marked vertices are collapsed into a single vertex.
    \textbf{(c)}~The result of applying $\insertG(\mathcal{S}, u, G_\nu, \phi)$ followed by an application of $\integV$ on the generated virtual vertices $v$ and $v'$.
  }
  \label{fig:insertGSPQR}
\end{figure}

We now introduce our first dynamic operation that allows us to actually change the represented graph by expanding a single vertex $u$ into an arbitrary connected graph $G_\nu$.
This is done by identifying $|N(u)|$ marked vertices in $G_\nu$ with the neighbors of $u$ via a bijection $\phi$ and then removing $u$ and its incident edges.
We use the ``occupied stars'' from the previous section to model the identification of these vertices, allowing us to defer the actual insertion to an application of  \integV.
We need to ensure that the inserted graph makes the same ``guarantees'' to the surrounding graph in terms of connectivity as the vertex it replaces,
that is all neighbors of $u$ (i.e. all marked vertices in $G_\nu$) need to be pairwise connected via paths in $G_\nu$ not using any other neighbor of $u$ (i.e. any other marked vertex).
Without this requirement, a single vertex could e.g. also be split into two non-adjacent halves, which could easily break a triconnected component apart.
Thus, we require $G_\nu$ to be biconnected when all marked vertices are collapsed into a single vertex.
Note that this also ensures that the old graph can be restored by contracting the vertices of the inserted graph.
For the sake of simplicity, we require vertex $u$ from the represented graph to have a single allocation vertex $v\in G_\mu$ with $\origV^{-1}(u)=\{v\}$ so that we only need to change a single allocation skeleton $G_\mu$ in the \skeldec.
As we will make clear later on, this condition can be satisfied easily.

Formally, let $u\in G_\mathcal{S}$ be a vertex that only has a single allocation vertex $v\in G_\mu$ (and thus only a single allocation skeleton $G_\mu$).
Let $G_\nu$ be an arbitrary, new graph containing $|N(u)|$ marked vertices, together with a bijection $\phi$ between the marked vertices in $G_\nu$ and the neighbors of $v$ in $G_\mu$.
We require $G_\nu$ to be biconnected when all marked vertices are collapsed into a single node.
Operation $\insertG(\mathcal{S}, u, G_\nu, \phi)$ yields an \eskeldec \eskeldecVars['] as follows, see also \Cref{fig:insertGSPQR}.
We interpret $G_\nu$ as skeleton and add it to $\mathcal G'$.
For each marked vertex $x$ in $G_\nu$, we set $\origV'(x)=\origV(\phi(x))$.
For all other vertices and edges in $G_\nu$, we set $\origV'$ and $\origE'$ to point to new vertices and edges forming a copy of $G_\nu$ in $\mathcal G_\mathcal{S'}$.
We connect every marked vertex in $G_\nu$ to a new virtual vertex $v'\in G_\nu$ using occupied edges.
We also convert $v$ to a virtual vertex, converting its incident edges to occupied edges while removing parallel edges.
Finally, we set $\twinV'(v)=v'$ and $\twinV'(v')=v$.

\begin{restatable}{lemma}{lemInsertG}\label{lem:insertG}
Applying $\insertG(\mathcal{S}, u, G_\nu, \phi)$ on an \eskeldec \eskeldecVars yields an \eskeldec \eskeldecVars['] with $G_\mathcal{S'}$ isomorphic to \replace{G_\mathcal{S}}{G_\nu}{u}{\phi}.
\end{restatable}
\newcommand{\proofInsertG}{
\begin{proof}
Condition \cond{bicon} remains satisfied, as the structure of $G_\mu$ remains unchanged and the resulting $G_\nu$ is biconnected by precondition.
Regarding $T_\mathcal{S}$, we are attaching a degree-1 node $\nu$ to an existing node $\mu$, thus condition \cond{tree} also remains satisfied.
As all vertices of $G_\nu$ except for the vertices in $N(v')$ got their new, unique copy assigned by $\origV'$ and $\origV'(N(v'))=\origV(N(v))$, condition \cond{orig-inj} is also satisfied for the new $G_\nu$.
As we updated $\origE$ alongside $\origV$ and $G_\nu$ contains no virtual edges, conditions \cond{orig-real} and \cond{orig-virt} remain satisfied.
As $\nu$ is a leaf of $T_\mathcal{S}$ with $\mu$ being its only neighbor, $\origV'(N(v'))\subset\origV(V_\mu)$, and $G_\nu$ is the only allocation skeleton for all vertices in $G_\nu\setminus N(v')$, condition \cond{subgraph} remains satisfied.
Conditions \cond{stars} and \cond{orig-stars} are satisfied by construction.
Finally, the mappings $\origE'$ and $\origV'$ are by construction updated to correctly reproduce the structure of $G_\nu$ in $G_\mathcal{S'}$.
\end{proof}
}
\ifthenelse{\boolean{long}}{\proofInsertG}{}

On its own, this operation is not of much use though, as graph vertices only rarely have a single allocation skeleton.
Furthermore, our goal is to dynamically maintain SPQR-trees, while this operation on its own will in most cases not yield an SPQR-tree.
To fix this, we introduce the full procedure $\insertGSPQR(\mathcal{S}, u, G_\nu, \phi)$ that can be applied to any graph vertex $u$ and that, given an SPQR-tree $\mathcal{S}$, yields the SPQR-tree of $\replace{G_\mathcal{S}}{G_\nu}{u}{\phi}$.
It consists of three preparations steps, the insertion of $G_\nu$, and two further clean-up steps:
\begin{enumerate}
\item \label[step]{step:pp-poly}
  We apply $\splitSP$ to each polygon allocation skeleton of $u$ with more than three vertices, using the neighbors of the allocation vertex of $u$ as separation pair.
\item \label[step]{step:pp-rigid}
  For each rigid allocation skeleton of $u$, we move the contained allocation vertex $v$ of $u$ to its own skeleton by applying $\isolV(\mathcal S, v)$.
\item \label[step]{step:pp-join}
  We exhaustively apply $\joinSP$ to any pair of allocation skeletons of $u$ that are adjacent in $T_\mathcal{S}$.
  Due to condition \cond{subgraph}, this yields a single component $G_\mu$ that is the sole allocation skeleton of $u$ with the single allocation vertex $v$ of $u$.
  Furthermore, the size of $G_\mu$ is linear in $\deg(u)$.
\item \label[step]{step:m-insert}
  We apply $\insertG$ to insert $G_\nu$ as skeleton, followed by an application of $\integV$ to the virtual vertices $\{v, v'\}$ introduced by the insertion, thus integrating $G_\nu$ into $G_\mu$. 
\item \label[step]{step:cu-spqr}
  We apply $\splitSP$ to all separation pairs in $G_\mu$ that do not involve a virtual vertex. These pairs can be found in linear time, e.g. by temporarily duplicating all virtual vertices and their incident edges and then computing the SPQR-tree.\footnote{Note that the wheels replacing virtual vertices in the proof of \Cref{lem:maintain-planar-rotation} also ensure this.}
\item \label[step]{step:cu-join}
  Finally, we exhaustively apply $\integV$ and also apply $\joinSP$ to any two adjacent polygons and to any two adjacent bonds to obtain the SPQR-tree of the updated graph.
\end{enumerate}

\begin{figure}[t]
  \begin{subfigure}[T]{0.5\textwidth}
    \centering
    \includegraphics[page=4,width=.95\linewidth]{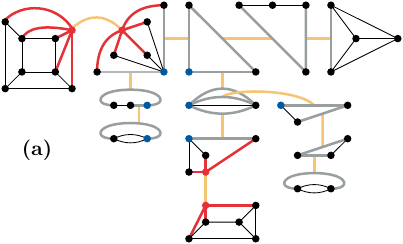}
    \phantomsubcaption
    \label{fig:eskeldec-split-tri}
  \end{subfigure}
  \hfill
  \begin{subfigure}[T]{0.5\textwidth}
    \centering
    \includegraphics[page=5,width=.95\linewidth]{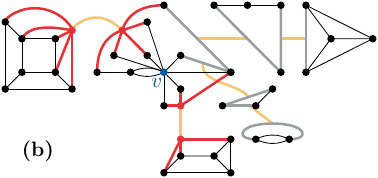}
    \phantomsubcaption
    \label{fig:eskeldec-join-alloc}
  \end{subfigure}

  \caption{
    The preprocessing steps of \insertGSPQR being applied to the SPQR-tree of \Cref{fig:skeldec-skeletons}.
    \textbf{(a)} The state after \cref{step:pp-rigid}, after all allocation skeletons of $u$ have been split.
    \textbf{(b)} The state after \cref{step:pp-join}, after all allocation skeletons of $u$ have been merged into a single one.
  }
  \label{fig:eskeldec}
\end{figure}

The basic idea behind the correctness of this procedure is that splitting the newly inserted component according to its SPQR-tree in \cref{step:cu-spqr} yields biconnected components that are each either a polygon, a bond, or ``almost'' triconnected.
The latter (and only those) might still contain virtual vertices and all their remaining separation pairs, which were not split in \cref{step:cu-spqr}, contain one of these virtual vertices.
This, together with the fact that there still may be pairs of adjacent skeletons where both are polygons or both are bonds, prevents the instance from being an SPQR-tree.
Both issues are resolved in \cref{step:cu-join}: The adjacent skeletons are obviously fixed by the \joinSP applications.
To show that the virtual vertices are removed by the \integV applications, making the remaining components triconnected, we need the following lemma.

\begin{lemma}\label{lem:integV-sep-pairs}
Let $G_\alpha$ be a triconnected skeleton containing a virtual vertex $v_\alpha$ matched with a virtual vertex $v_\beta$ of a biconnected skeleton $G_\beta$.
Furthermore, let $P\subseteq \binom{V(G_\beta)}{2}$ be the set of all separation pairs in $G_\beta$.
An application of $\integV(\mathcal S, (v_\alpha, v_\beta))$ yields a biconnected skeleton $G_\mu$ with separation pairs $P'=\{\{u,v\}\in P\mid v_\beta\notin\{u,v\}\}$.
\end{lemma}
\begin{proof}

We partition the vertices of $G_\mu$ into the sets $A,B$, and $N$ depending on whether the vertex stems from $G_\alpha$, $G_\beta$, or both, respectively.
The set $N$ thus contains the neighbors of $v_\alpha$, which were identified with the neighbors of $v_\beta$.
%
We will now show by contradiction that $G_\mu$ contains no separation pairs except for those in $P'$.
Thus, consider a separation pair $u,v\in G_\mu$ not in $P'$.
%
First, consider the case where $u,v\in A\cup N$.
Observe that removing $u,v$ in this case leaves $B$ connected.
Thus, we can contract all vertices of $B$ into a single vertex, reobtain~$G_\alpha$ and see that $u,v$ is a separation pair in $G_\alpha$.
This contradicts the precondition that $G_\alpha$ is triconnected.
Now consider the case where $u,v\in B\cup N$.
Analogously to above, we find that $u,v$ is a separation pair in $G_\beta$ that does not contain $v_\beta$, a contradiction to $\{u,v\} \notin P'$.
Finally, consider the remaining case where, without loss of generality, $u\in A, v\in B$.
Since $\{u,v\}$ is a separation pair, $u$ has two neighbors $x,y$ that lie in different connected components of $G_\mu-\{u,v\}$ and therefore also in different components of $(G_\mu-\{u,v\})-B$ which is isomorphic to $G_\alpha-\{u,v_\alpha\}$.
This again contradicts the precondition that $G_\alpha$ is triconnected.
\end{proof}


\begin{theorem}\label{thm:insertGSPQR}
Applying $\insertGSPQR(\mathcal{S}, u, G_\nu, \phi)$ to an SPQR-tree $\mathcal{S}$ yields an SPQR-tree $\mathcal{S}'$ in $O(|G_\nu|)$ time with $G_\mathcal{S'}$ isomorphic to \replace{G_\mathcal{S}}{G_\nu}{u}{\phi}.
\end{theorem}
\begin{proof}
As all operations that are applied leave the \eskeldec valid, the final \eskeldec $\mathcal{S}'$ is also valid.
Observe that the purpose of the preprocessing \cref{step:pp-poly,step:pp-join,step:pp-rigid} is solely to ensure that the preconditions of \insertG are satisfied and the affected component is not too large.
Note that all rigids split in \cref{step:pp-rigid} remain structurally unmodified in the sense that edges only changed their type, but the graph and especially its triconnectedness remains unchanged.
\Cref{step:m-insert} performs the actual insertion and yields the desired represented graph according to \Cref{lem:insertG}.
It thus remains to show that the clean-up steps turn the obtained \eskeldec into an SPQR-tree.
Applying $\integV$ exhaustively in \cref{step:cu-join} ensures that the \eskeldec is equivalent to a non-extended one (\Cref{obs:ne-skeldec}).
Recall that a non-extended \skeldec is an SPQR-tree if all skeletons are either polygons, bonds or triconnected and two adjacent skeletons are never both polygons or both bonds (\Cref{thm:spqrt}).
\Cref{step:cu-join} ensures that the second half holds, as joining two polygons (or two bonds) with $\joinSP$ yields a bigger polygon (or bond, respectively).
Before \cref{step:cu-join}, all skeletons that are not an allocation skeleton of $u$ are still unmodified and thus already have a suitable structure, i.e., they are either polygons, bonds or triconnected.
Furthermore, the allocation skeletons of $u$ not containing virtual vertices also have a suitable structure, as their splits were made according to the SPQR-tree in \cref{step:cu-spqr}.
It remains to show that the remaining skeletons, that is those resulting from the $\integV$ applications in \cref{step:cu-join}, are triconnected.
Note that in these skeletons, \cref{step:cu-spqr} ensures that every separation pair consists of at least one virtual vertex, as otherwise the computed SPQR-tree would have split the skeleton further.
Further note that, for each of these virtual vertices, the matched partner vertex is part of a structurally unmodified triconnected skeleton that was split in \cref{step:pp-rigid}.
\Cref{lem:integV-sep-pairs} shows that applying $\integV$ does not introduce new separation pairs while removing two virtual vertices if one of the two sides is triconnected.
We can thus exhaustively apply $\integV$ and thereby remove all virtual vertices and thus also all separation pairs, obtaining triconnected components.
This shows that the criteria for being an SPQR-tree are satisfied and, as $\insertG$ expanded $u$ to $G_\nu$ in the represented graph, we now have the unique SPQR-tree of $\replace{G_\mathcal{S}}{G_\nu}{u}{\phi}$.

Note that all operations we used can be performed in time linear in the degree of the vertices they are applied on.
For the bipartition of bridges input to \splitSP, it is sufficient to describe each bridge via its edges incident to the separation pair instead of explicitly enumerating all in vertices in the bridge.
Thus, the applications of \splitSP and \isolV in \cref{step:pp-poly,step:pp-rigid} touch every edge incident to $u$ at most once and thus take $O(\deg(u))$ time.
Furthermore, they yield skeletons that have a size linear in the degree of their respective allocation vertex of $u$.
As the subtree of $u$'s allocation skeletons has size at most $\deg(u)$, the \joinSP applications of \cref{step:pp-join} also take at most $O(\deg(u))$ time.
It also follows that the resulting single allocation skeleton of $u$ has size $O(\deg(u))$.
The applications of \insertG and \integV in \cref{step:m-insert} can be done in time linear in the number of identified neighbors, which is $O(\deg(u))$.
Generating the SPQR-tree of the inserted graph in \cref{step:cu-spqr} (where all virtual vertices where replaced by wheels) can be done in time linear in the size of the inserted graph~\cite{Gutwenger2001,Hopcroft1973}, that is $O(|G_\nu|)$.
Applying \splitSP according to all separation pairs identified by this SPQR-tree can also be done in $O(|G_\nu|)$ time in total.
Note that there are at most $\deg(u)$ edges between the skeletons that existed before \cref{step:m-insert} and those that were created or modified in \cref{step:m-insert,step:cu-spqr}, and these are the only edges that might now connect two polygons or two bonds.
As these tree edges have one endpoint in the single allocation skeleton of $u$, the applications of $\integV$ and $\joinSP$ in \cref{step:cu-join} run in $O(\deg(u))$ time in total.
Furthermore, they remove all pairs of adjacent polygons and all pairs of adjacent bonds.
This shows that all steps take $O(\deg(u))$ time, except for \cref{step:cu-spqr}, which takes $O(|G_\nu|)$ time.
As the inserted graph contains at least one vertex for each neighbor of $u$, the total runtime is in $O(|G_\nu|)$.
\end{proof}

\begin{corollary}\label{cor:merge-spqr-trees}
Let $\mathcal S_1,\mathcal S_2$ be two SPQR-trees together with vertices $u_1\in G_{\mathcal S_1},$ $u_2\in G_{\mathcal S_2}$, and let $\phi$ be a bijection between the edges incident to $u_1$ and the edges incident to $u_2$.
Operation $\mergeSPQR(\mathcal S_1,\mathcal S_2,u_1,u_2,\phi)$ yields the SPQR-tree of the graph \replace{G_{\mathcal S_1}}{G_{\mathcal S_2}-u_2}{u_1}{\phi}, i.e. the union of both graphs where the edges incident to $u_1,u_2$ were identified according to $\phi$ and $u_1,u_2$ removed, in time $O(\deg(u_1))=O(\deg(u_2))$.
\end{corollary}
\begin{proof}
Operation \mergeSPQR works similar to the more general \insertGSPQR, although the running time is better because we already know the SPQR-tree for the graph being inserted.
We apply the preprocessing \cref{step:pp-poly,step:pp-rigid,step:pp-join} to ensure that both $u_1$ and $u_2$ have sole allocation vertices $v_1$ and $v_2$, respectively.
To properly handle parallel edges, we subdivide all edges incident to $u_1, u_2$ (and thus also the corresponding real edges incident to $v_1,v_2$) and then identify the subdivision vertices of each pair of edges matched by $\phi$.
By deleting vertices $v_1$ and $v_2$ and suppressing the subdivision vertices (that is, removing them and identifying each pair of incident edges) we obtain a skeleton $G_\mu$ that has size $O(\deg(u_1))=O(\deg(u_2))$.
Finally, we apply the clean-up \cref{step:cu-join,step:cu-spqr} to $G_\mu$ to obtain the final SPQR-tree.
Again, as the partner vertex of every virtual vertex in the allocation skeletons of $u$ is part of a triconnected skeleton, applying \integV exhaustively in \cref{step:cu-join} yields triconnected skeletons.
As previously discussed, the preprocessing and clean-up steps run in time linear in degree of the affected vertices, thus the overall runtime is $O(\deg(u_1))=O(\deg(u_2))$ in this~case.
\end{proof}


\subsection{Maintaining Planarity and Vertex Rotations}\label{sec:maintenance}
Note that expanding a vertex of a planar graph using another planar graph using \insertGSPQR (or merging two SPQR-trees of planar graphs using \Cref{cor:merge-spqr-trees}) might actually yield a non-planar graph.
This is, e.g., because the rigids of both graphs might require incompatible orders for the neighbors of the replaced vertex.
The aim of this section is to efficiently detect this case, that is a planar graph turning non-planar.
To check a general graph for planarity, it suffices to check the rigids in its SPQR-tree for planarity and each rigid allows exactly two planar embeddings, where one is the reverse of the other~\cite{Battista1996p}.
Thus, if a graph becomes non-planar through an application of \insertGSPQR, this will be noticeable from the triconnected allocation skeletons of the replaced vertex.
To be able to immediately report if the instance became non-planar, we need to maintain a rotation, that is a cyclic order of all incident edges, for each vertex in any triconnected skeleton. 
Note that we do not track the direction of the orders, that is we only store the order up to reversal.
As discussed later, the exact orders can also be maintained with a slight overhead.

\begin{theorem}\label{lem:maintain-planar-rotation}
SPQR-trees support the following operations:
\begin{itemize}
 \item $\normalfont\insertGSPQR(\mathcal{S}, u, G_\nu, \phi)$: expansion of a single vertex $u$ in time $O(|G_\nu|)$,
 \item $\normalfont\mergeSPQR(\mathcal S_1,\mathcal S_2,u_1,u_2,\phi)$: merging of two SPQR-trees in time $O(\deg(u_1))$,
 \item $\normalfont\texttt{IsPlanar}$: queries whether the represented graph is planar in time $O(1)$, and
 \item $\normalfont\texttt{Rotation}(u)$: queries for one of the two possible rotations of vertices $u$ in planar triconnected skeletons in time $O(1)$.
\end{itemize}
\end{theorem}
\begin{proof}
Note that the boolean flag \texttt{IsPlanar} together with the \texttt{Rotation} information can be computed in linear time when creating a new SPQR-tree and that expanding a vertex or merging two SPQR-trees cannot turn a non-planar graph planar.
We make the following changes to the operations \insertGSPQR and \mergeSPQR described in \Cref{thm:insertGSPQR} and \Cref{cor:merge-spqr-trees} to maintain the new information.
After a triconnected component is split in \cref{step:pp-rigid} we now introduce further structure to ensure that the embedding is maintained on both sides.
The occupied edges generated around the split-off vertex $v$ (and those around its copy $v'$) are subdivided and connected cyclically according to \texttt{Rotation}$(v)$.
Instead of ``stars'', we thus now generate occupied ``wheels'' that encode the edge ordering in the embedding of the triconnected component.
When generating the SPQR-tree of the modified subgraph in \cref{step:cu-spqr}, now containing occupied wheels instead of only stars, we also generate a planar embedding for all its triconnected skeletons.
If no planar embedding can be found for at least one skeleton, we report that the resulting instance is non-planar by setting \texttt{IsPlanar} to false.
Otherwise, after performing all splits indicated by the SPQR-tree, we assign \texttt{Rotation} by generating embeddings for all new rigids.
Note that for all skeletons with virtual vertices, the generated embedding will be compatible with the one of the neighboring triconnected component, that is, the rotation of each virtual vertex will line up with that of its matched partner vertex, thanks to the inserted wheel.
Finally, before applying $\integV$ in \cref{step:cu-join}, we contract each occupied wheel into a single vertex to re-obtain occupied stars.
The creation and contraction of wheels adds an overhead that is at most linear in the degree of the expanded vertex and the generation of embeddings for the rigids can be done in time linear in the size of the rigid.
Thus, this does not affect the asymptotic runtime of both operations.
\end{proof}

\begin{restatable}{corollary}{corUnionFind}
The data structure from \Cref{lem:maintain-planar-rotation} can be adapted to also provide the exact rotations with matching direction for every vertex in a rigid.
Furthermore, it can support queries whether two vertices $v_1,v_2$ are connected by at least 3 different vertex-disjoint paths via $\texttt{3Paths}(v_1,v_2)$ in $O((\deg(v_1)+\deg(v_2))\cdot\alpha(n))$ time.
These adaptions change the runtime of $\insertGSPQR$ to $O(\deg(u)\cdot\alpha(n)+|G_\nu|)$, that of $\mergeSPQR$ to $O(\deg(u_1)\cdot\alpha(n))$, and that of $\texttt{Rotation}(u)$ to $O(\alpha(n))$.
\end{restatable}
\newcommand{\proofUnionFind}{
\begin{proof}
The exact rotation information for $\texttt{Rotation}$ can be maintained by using union-find to keep track of the rigid a vertex belongs to and synchronizing the reversal of all vertices within one rigid when two rigids are merged by \integV as follows.
We create a union-find set for every vertex in a triconnected component and apply \texttt{Union} to all vertices in the same rigid.
Next to the pointer indicating the representative in the union-find structure, we store a boolean flag indicating
whether the rotation information for the current vertex is reversed with regard to rotation of its direct representative.
To find whether a $\texttt{Rotation}$ needs to be flipped, we accumulate all flags along the path to the actual representative of a vertex by using an exclusive-or.
As $\texttt{Rotation}(u)$ thus relies on the \texttt{Find} operation, its amortized runtime is $O(\alpha(n))$.
When merging two rigids with \integV, we also perform a \texttt{Union} on their respective representatives (which we need to \texttt{Find} first), making $\integV(\mathcal S, (v_\alpha, v_\beta))$ run in $O(\deg(v_\alpha)+\alpha(n))$.
We also compare the $\texttt{Rotation}$ of the replaced vertices and flip the flag stored with the vertex that does not end up as the representative if they do not match.
In total, this makes $\insertGSPQR$ run in $O(\deg(u)\cdot\alpha(n)+|G_\nu|)$ time as there can be up to $\deg(u)$ split rigids.
Furthermore, $\mergeSPQR$ now runs in $O(\deg(u_1)\cdot\alpha(n))$ time.

Maintaining the information in which rigid a skeleton vertex is contained in can then also be used to answer queries whether two arbitrary vertices are connected by three disjoint paths.
This is exactly the case if they are part of the same rigid, appear as poles of the same bond or are connected by a virtual edge in a polygon.
This can be checked by enumerating all allocation skeletons of both vertices, which can be done in time linear in their degree.
As finding each of the skeletons may require a \texttt{Find} call, the total runtime for this is in $O((\deg(v_1)+\deg(v_2))\cdot\alpha(n))$.
\end{proof}
}
\ifthenelse{\boolean{long}}{\proofUnionFind}{}

\ifthenelse{\boolean{long}}{
  \section{Application to Synchronized Planarity}\label{sec:application}
  In this section, we will give some background on the historical development of and further details on the problems \cplan and \pqplan together with summary of the algorithm of Bläsius et al. for solving both problems.
  Furthermore, we will show how our and also previous work on dynamic SPQR-trees can be used in the context of both problems.

  \subsection{Background and Discussion}

\label{sec:background}

Lengauer~\cite{Lengauer1989} first discussed \cplan under a different name in 1989, which is why it was later independently rediscovered by Feng et al.~\cite{Feng1995} in 1995.
Both gave polynomial-time algorithms for the case where the subgraph induced by any cluster is connected.
In contrast, the question whether the general problem with disconnected clusters allows an efficient solution remained open for 30 years.
In that time, polynomial-time algorithms were found for many special-cases~\cite{Angelini2019,Cortese2008,Fulek2015,Gutwenger2002} before Fulek and Tóth~\cite{Fulek2019} found an $O((n+d)^8)$ solution in 2019.
Shortly thereafter, Bläsius et al.~\cite{Blaesius2021} gave a solution with runtime in $O((n+d)^2)$ that also exposes the main concepts needed to solve \cplan.
The solution works via a linear-time reduction to the problem \pqplan, for which Bläsius et al. gave a quadratic algorithm.
We improve the runtime of the latter algorithm.
As \pqplan can be used as a modeling tool for several other constrained planarity problems next to \cplan~\cite{Blaesius2021}, this also improves the time needed for solving any constrained planarity problem that can be solved via a linear-time reduction to \pqplan; see \Cref{tab:constplanprobs}.

In \cplan, the embedding has to respect a laminar family of clusters~\cite{Blaesius2016,Lengauer1989}, that is every vertex is part of some (hierarchically nested) cluster and an edge may only cross a cluster boundary if it connects a vertex from the inside with one from the outside.
In \pqplan, we are given a matching on some of the vertices in the graph and seek an embedding such that the rotations matched vertices line up under a given bijection~\cite{Blaesius2021}.
The synchronization constraints imposed by matching two vertices are also called \emph{pipe}. 
The reduction from the former problem to the latter employs the CD-tree representation of \cplan~\cite{Blaesius2016}, where each cluster is represented as individual skeleton in which adjacent clusters were collapsed into single ``virtual vertices''.
The order of the edges ``leaving'' one cluster via a virtual vertex now needs to line up with the order in which they ``enter'' an adjacent cluster via its corresponding virtual vertex (see also~\cite[Figure 6]{Blaesius2021}).

\begin{figure}[t]
  \centering
  \includegraphics[page=1]{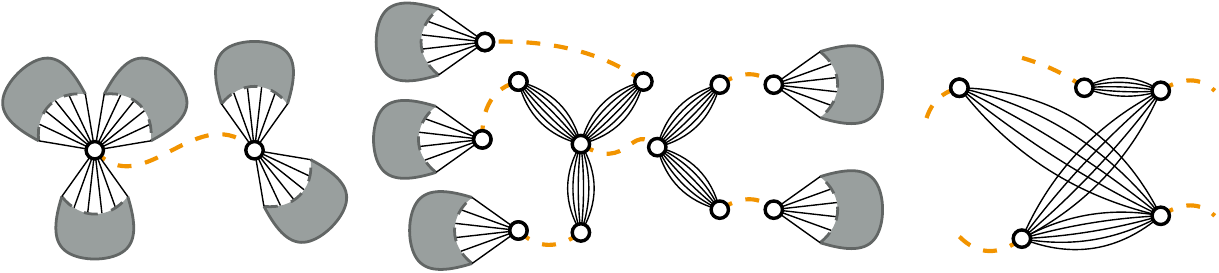}\\[.25cm]
  \includegraphics[page=2]{graphics/syncplan}\\[.2cm]
  \includegraphics[page=3]{graphics/syncplan}
  \caption{
    Schematic representation of the three operations used by Bläsius et al.~\cite{Blaesius2021} for solving \pqplan. Matched vertices are shown as bigger disks, the matching is indicated by the orange dotted lines.
    \textbf{Top:}
      Two cut-vertices matched with each other (left), the result of encapsulating their incident blocks (middle) and the bipartite graph resulting from joining both cut-vertices (right).
    \textbf{Middle:}
      A matched non-cut-vertex with a non-trivial embedding tree (left) that is propagated to replace both the vertex and its partner (right).
    \textbf{Bottom:}
      Three different cases of matched vertices with trivial embedding trees (blue) and how their pipes can be removed or replaced (red).
  }
  \label{fig:syncplan-ops}
\end{figure}

The algorithm for solving \pqplan works by removing an arbitrary pipe each step, using one of three operations depending on the graphs around the matched vertices, see \Cref{fig:syncplan-ops}.
\begin{description}
\item[\contract]
If both vertices of the pipe are cut-vertices, they are ``encapsulated'' by taking a copy of their respective components and then collapsing each incident block to a single vertex to obtain stars with matched centers that have multiple parallel edges connecting them to their ray vertices.
The original cut-vertices are split up so that each incident block gets its own copy and these copies are synchronized with the respective vertex representing a collapsed block.
Now the cut-vertices can be removed by ``joining'' both stars, that is identifying their incident edges according to the bijection that is given alongside the matching.
\item[\propagate]
If one of the vertices is not a cut-vertex and has an embedding tree that not only consists of a single P-node, two copies of this embedding tree are inserted (``propagated'') in place of both matched vertices, respectively.
The inner nodes of the embedding trees are synchronized by matching corresponding vertices.
\item[\simplify]
In the remaining case, one of the vertices is not a cut-vertex but has a trivial embedding tree, i.e., only appears in a single parallel skeleton and no rigid skeleton in the SPQR-tree.
If the vertex (or, more precisely, the parallel that completely defines it rotation) can respect arbitrary rotations, we can simply remove the pipe.
The only exception to this is when the other pole of the parallel is also matched, in which case we can ``short-circuit'' the matching across the parallel.
\end{description}

To summarize, every operation removes a pipe from the matching, while potentially introducing new pipes with vertices that have a smaller degree.
Using a potential function, it can be shown that the progress made by the removal always dominates overhead of the newly-introduced pipes, and that the operations needed to remove all pipes is limited by the total degree of all matched vertices.
Furthermore, the resulting instance without pipes can be solved in linear time.
All of the three operations run in time linear in the degree of the un-matched vertices if the embedding trees they depend on are available.
The contribution of this paper is to efficiently provide the embedding trees, which would require processing entire connected components at each step when done naïvely.
Using the fully-dynamic SPQR-tree by Holm and Rotenberg~\cite{Holm2020p,Holm2020}, this can be achieved with a poly-log cost of $O(\Delta \cdot \log^3 n)$ leading to an overall runtime of $O(m \cdot \Delta \cdot \log^3 n)$.
Using the node expansion from this paper, we can improve the runtime from spending time linear in the size of the input instance ($O(m)$) for each of the linearly many operations, to only spending time linear in the maximum degree ($O(\Delta)$) on each operation.
The reduction from \cplan creates an instance of size $O(n+d)$ in which the total degree of matched vertices is in $O(d)$, corresponding to the total number of times an edge crosses a cluster boundary.
Note that, while this means that $O(d)$ operations are sufficient to reach a reduced instance, the number of crossings between edges and cluster boundaries can be quadratic in the number of vertices in a planar graph.
We also note that while the improvement over using the Holm and Rotenberg approach is only poly-logarithmic, our datastructure has the additional benefit of being conceptually simpler and thus also more likely to improve performance in practice.

  \subsection{Using Node Expansion for Solving Synchronized Planarity}
}{
  \section{Application to Synchronized Planarity}\label{sec:application}
}

We show how {\eskeldec}s and their dynamic operation \insertGSPQR can be used to improve the runtime of the algorithm for solving \pqplan by Bläsius et al.~\cite{Blaesius2021} from $O(m^2)$ to $O(m\cdot\Delta)$, where $\Delta$ is the maximum pipe degree.
\ifthenelse{\boolean{long}}{As already explained in the previous section, the}{The} algorithm spends a major part of its runtime on computing so-called embedding trees, which describe all possible rotations of a single vertex in a planar graph and are used to communicate embedding restrictions between vertices with synchronized rotation.
Once the embedding trees are available, the at most $O(m)$ executed operations run in time linear in the degree of the pipe/vertex they are applied on, that is in $O(\Delta)$~\cite{Blaesius2021}.
Thus, being able to generate these embedding trees efficiently by maintaining the SPQR-trees they are derived from is our main contribution towards the speedup of the \pqplan algorithm.
\ifthenelse{\boolean{long}}{}{%
  See \Cref{sec:background} for more details on the problems \textsc{Syn}\-\textsc{chro}\-\textsc{nized} and \cplan and their solution.
  There, we also give a short overview over the operations Bläsius et al.~\cite{Blaesius2021} use for solving \pqplan, which we improve in the proof of \Cref{thm:SyncPlan}.
}

\label{sec:embed-trees}

An \emph{embedding tree} $\mathcal T_v$ for a vertex $v$ of a biconnected graph $G$ describes the possible cyclic orderings or \emph{rotations} of the edges incident to $v$ in all planar embeddings of $G$~\cite{Booth1976}.
The leaves of $\mathcal T_v$ are the edges incident to $v$, while its inner nodes are partitioned into two categories:
\emph{Q-nodes} define an up-to-reversal fixed rotation of their incident tree edges, while \emph{P-nodes} allow arbitrary rotation; see \Cref{fig:embedding-tree}.
To generate the embedding tree we use the observation about the relationship of SPQR-trees and embedding trees described by Bläsius and Rutter~\cite[Section 2.5]{Blaesius2011}:
there is a bijection between the P- and Q-nodes in the embedding tree of $v$ and the bond and triconnected allocation skeletons of $v$ in the SPQR-tree of $G$, respectively.
\ifthenelse{\boolean{long}}{}{Note that the detailed constructions for the following statements are given in the respective proofs the appendix.}

\begin{restatable}{lemma}{lemEmbedTrees}\label{lem:embedTrees}
Let $\mathcal S$ be an SPQR-tree with a planar represented graph $G_{\mathcal S}$.
The embedding tree for a vertex $v\in G_{\mathcal S}$ can be found in time $O(\deg(v))$.
\end{restatable}
\newcommand{\proofEmbedTrees}{
\begin{proof}
We use the rotation information from \Cref{lem:maintain-planar-rotation} and furthermore maintain an (arbitrary) allocation vertex for each vertex in $G_{\mathcal S}$.
To compute the embedding tree of a vertex $v$ starting at the allocation vertex $u$ of $v$, we will explore the SPQR-tree by using $\twinE$ on one of the edges incident to $u$ and then finding the next allocation vertex of $v$ as one endpoint of the obtained edge.
If $u$ has degree 2, it is part of a polygon skeleton that does not induce a node in the embedding tree.
We thus move on to its neighboring allocation skeletons and will also similarly skip over any other polygon skeleton we encounter.
If $u$ has degree 3 or greater, we inspect two arbitrary incident edges:
if they lead to the same vertex, $u$ is the pole of a bond, and we generate a P-node.
Otherwise it is part of a triconnected component, and we generate a Q-node.
We now iterate over the edges incident to $u$, in the case of a triconnected component using the order given by the rotation of $u$.
For each real edge, we attach a corresponding leaf to the newly generated node.
The graph edge corresponding to the leaf can be obtained from $\origE$.
For each virtual edge, we recurse on the respective neighboring skeleton and attach the recursively generated node to the current node.
As $u$ can only be part of $\deg(u)$ many skeletons, which form a subtree of $T_\mathcal{S}$, and the allocation vertices of $u$ in total only have $O(\deg(u))$ many virtual and real edges incident, this procedure yields the embedding tree of $u$ in time linear in its degree.
\end{proof}
}
\ifthenelse{\boolean{long}}{\proofEmbedTrees}{}


Our data structure can now be used to reduce the runtime of solving \pqplan by generating an SPQR-tree upfront, maintaining it throughout all applied operations, and deriving any needed embedding tree from the SPQR-tree.
\begin{restatable}{theorem}{thmSyncPlan}\label{thm:SyncPlan}
\pqplan can be solved in time in $O(m\cdot\Delta)$, where $m$ is the number of edges and $\Delta$ is the maximum degree of a pipe.
\end{restatable}
\newcommand{\proofSyncPlan}{
\begin{proof}
The algorithm works by splitting the pipes representing synchronization constraints until they are small enough to be trivial.
It does so by exhaustively applying the three operations \contract, \propagate and \simplify depending on the graph structure around the pairs of synchronized vertices.
As mentioned by Bläsius et al., all operations run in time linear in the degree of the pipe they are applied on if the used embedding trees are known, and $O(m)$ operations are sufficient to solve a given instance~\cite{Blaesius2021}.
Our modification is that we maintain an SPQR-tree for each biconnected component and then generate the needed embedding trees on-demand in linear time using \Cref{lem:embedTrees}.
See \Cref{sec:background} for more background on the \pqplan operations modified in the following.

Operation \simplify can be applied if the graph around a synchronized vertex $v$ allows arbitrary rotations of $v$, that is the embedding tree of $v$ is trivial.
In this case, the pipe can be removed without modifying the graph structure.
Thus, we can now easily check the preconditions of this operations without making any changes to the SPQR-tree.

\propagate takes the non-trivial embedding tree of one synchronized vertex $v$ and inserts copies of the tree in place of $v$ and its partner, respectively.
Synchronization constraints on the inner vertices of the inserted trees are used to ensure that they are embedded in the same way.
We use \insertGSPQR to also insert the embedding tree into the respective SPQR trees, representing Q-nodes using wheels.
When propagating into a cutvertex we also need to check whether two or more incident blocks merge.
We form equivalence classes on the incident blocks, where two blocks are in the same class if
1) the two subtrees induced by their respective edges share at least two nodes
2) both induced subtrees share a C-node that has degree at least 2 in both subtrees.
Blocks in the same equivalence class will end up in the same biconnected component as follows:
We construct the subtree induced by all edges in the equivalence class and add a single further node for each block in the class, connecting all leaves to the node of the block the edges they represent lead to.
We calculate the SPQR-tree for this biconnected graph and then merge the SPQR-trees of the individual blocks into it by applying \Cref{cor:merge-spqr-trees}.
As \insertGSPQR (and similarly all \mergeSPQR applications) runs in time linear in the size of the inserted PQ-tree, which is limited by the degree of the vertex it represents, this does not negatively impact the running time of the operation.

Operation \contract generates a new bipartite component representing how the edges of the blocks incident to two synchronized cutvertices are matched with each other. 
The size of this component is linear in the degree of the synchronized vertices.
Thus, we can freshly compute the SPQR-tree for the generated component in linear time, which also does not negatively impact the running time.

Furthermore, as we now no longer need to iterate over whole connected components to generate the embedding trees, we are also no longer required to ensure those components do not grow to big.
We can thus also directly contract pipes between two distinct biconnected components using \Cref{cor:merge-spqr-trees} instead of having to insert PQ-trees using \propagate.
This may improve the practical runtime, as \propagate might require further operations to clean up the generated pipes, while the direct contraction entirely removes a pipe without generating new ones.
\end{proof}
}
\ifthenelse{\boolean{long}}{\proofSyncPlan}{
\begin{proof}[Proof Sketch.]
See \Cref{sec:background} for more background on the \pqplan operations modified in the following.
Operation \propagate expands a vertex into a tree corresponding to the embedding tree of its partner vertex with synchronized rotation.
This expansion can also be done in the SPQR-tree without a runtime overhead, while some care needs to be taken when expanding cut-vertices, as different parts of the tree need to be expanded in different blocks.
Operation \contract generates a new bipartite component linear in size to the pipe it removes.
Thus, the SPQR-tree for this new component can be computed without a runtime overhead.
All other operations do not affect the SPQR-tree and once embedding trees are available, of the at most $O(m)$ applied operations, each takes $O(\Delta)$ time~\cite{Blaesius2021}.
\end{proof}
}

\begin{restatable}{corollary}{corCPlan}
\cplan can be solved in time in $O(n+d\cdot \Delta)$, where $d$ is the total number of crossings between cluster borders and edges and $\Delta$ is the maximum number of edge crossings on a single cluster border.
\end{restatable}
\newcommand{\proofCPlan}{
\begin{proof}
Note that for a graph not containing parallel edges to be planar, the number of edges has to be linear in the number of vertices.
We apply the reduction from \cplan to \pqplan as described by Bläsius et al.~\cite{Blaesius2021}.
Ignoring the parallel edges generated by the CD-tree, we can generate an SPQR-tree for every component of the resulting instance in $O(n)$ time in total.
The instance contains one pipe for every cluster boundary, where the degree of a pipe corresponds to the number of edges crossing the respective cluster boundary.
Thus, the potential described by Bläsius et al.~\cite{Blaesius2021}, which sums up the degrees of all pipes with a constant factor depending on the endpoints of each pipe, is in $O(d)$.
Each operation applied when solving the \pqplan instance runs in time $O(\Delta)$ (the maximum degree of a pipe) and reduces the potential by at least 1.
Thus, a reduced instance without pipes, which can be solved in linear time, can be reached in $O(d\cdot\Delta)$ time.
\end{proof}
}
\ifthenelse{\boolean{long}}{\proofCPlan}{}



\bibliography{bibliography}

\begin{thebibliography}{10}

\bibitem{Angelini2014}
P.~Angelini, T.~Bl{\"{a}}sius, and I.~Rutter.
\newblock Testing mutual duality of planar graphs.
\newblock {\em International Journal of Computational Geometry \&
  Applications}, 24(4):325--346, 2014.
\newblock \href {http://arxiv.org/abs/1303.1640} {\path{arXiv:1303.1640}},
  \href {https://doi.org/10.1142/S0218195914600103}
  {\path{doi:10.1142/S0218195914600103}}.

\bibitem{Angelini2019}
P.~Angelini and G.~{Da Lozzo}.
\newblock Clustered planarity with pipes.
\newblock {\em Algorithmica}, 81(6):2484--2526, 2019.
\newblock \href {https://doi.org/10.1007/s00453-018-00541-w}
  {\path{doi:10.1007/s00453-018-00541-w}}.

\bibitem{Angelini2009}
P.~Angelini, G.~{Di Battista}, and M.~Patrignani.
\newblock Finding a minimum-depth embedding of a planar graph in {$O(n^4)$}
  time.
\newblock {\em Algorithmica}, 60(4):890--937, 2009.
\newblock \href {https://doi.org/10.1007/s00453-009-9380-6}
  {\path{doi:10.1007/s00453-009-9380-6}}.

\bibitem{Angelini2016a}
P.~Angelini, G.~D. Lozzo, G.~{Di Battista}, and F.~Frati.
\newblock Strip planarity testing for embedded planar graphs.
\newblock {\em Algorithmica}, 77(4):1022--1059, 2016.
\newblock \href {https://doi.org/10.1007/s00453-016-0128-9}
  {\path{doi:10.1007/s00453-016-0128-9}}.

\bibitem{Biedl1997}
T.~C. Biedl, G.~Kant, and M.~Kaufmann.
\newblock On triangulating planar graphs under the four-connectivity
  constraint.
\newblock {\em Algorithmica}, 19(4):427--446, 1997.
\newblock \href {https://doi.org/10.1007/PL00009182}
  {\path{doi:10.1007/PL00009182}}.

\bibitem{Bienstock1989}
D.~Bienstock and C.~L. Monma.
\newblock Optimal enclosing regions in planar graphs.
\newblock {\em Networks}, 19(1):79--94, 1989.
\newblock \href {https://doi.org/10.1002/net.3230190107}
  {\path{doi:10.1002/net.3230190107}}.

\bibitem{Bienstock1990}
D.~Bienstock and C.~L. Monma.
\newblock On the complexity of embedding planar graphs to minimize certain
  distance measures.
\newblock {\em Algorithmica}, 5(1):93--109, 1990.
\newblock \href {https://doi.org/10.1007/bf01840379}
  {\path{doi:10.1007/bf01840379}}.

\bibitem{Blaesius2021}
T.~Bl{\"{a}}sius, S.~D. Fink, and I.~Rutter.
\newblock Synchronized planarity with applications to constrained planarity
  problems.
\newblock In {\em Proceedings of the 29th Annual European Symposium on
  Algorithms (ESA'21)}, volume 204 of {\em LIPIcs}, pages 19:1--19:14, 2021.
\newblock \href {https://doi.org/10.4230/LIPIcs.ESA.2021.19}
  {\path{doi:10.4230/LIPIcs.ESA.2021.19}}.

\bibitem{Blaesius2016}
T.~Bl{\"{a}}sius and I.~Rutter.
\newblock A new perspective on clustered planarity as a combinatorial embedding
  problem.
\newblock {\em Theoretical Computer Science}, 609:306--315, 2016.
\newblock \href {http://arxiv.org/abs/1506.05673} {\path{arXiv:1506.05673}},
  \href {https://doi.org/10.1016/j.tcs.2015.10.011}
  {\path{doi:10.1016/j.tcs.2015.10.011}}.

\bibitem{Blaesius2011}
T.~Bl{\"{a}}sius and I.~Rutter.
\newblock Simultaneous {PQ}-ordering with applications to constrained embedding
  problems.
\newblock {\em {ACM} Transactions on Algorithms}, 12(2):16:1--16:46, 2016.
\newblock \href {https://doi.org/10.1145/2738054} {\path{doi:10.1145/2738054}}.

\bibitem{Blaesius2016a}
T.~Bl{\"{a}}sius, I.~Rutter, and D.~Wagner.
\newblock Optimal orthogonal graph drawing with convex bend costs.
\newblock {\em {ACM} Transactions on Algorithms}, 12(3):33:1--33:32, 2016.
\newblock \href {https://doi.org/10.1145/2838736} {\path{doi:10.1145/2838736}}.

\bibitem{Booth1976}
K.~S. Booth and G.~S. Lueker.
\newblock Testing for the consecutive ones property, interval graphs, and graph
  planarity using {PQ}-tree algorithms.
\newblock {\em Journal of Computer and System Sciences}, 13(3):335--379, 1976.
\newblock \href {https://doi.org/10.1016/s0022-0000(76)80045-1}
  {\path{doi:10.1016/s0022-0000(76)80045-1}}.

\bibitem{Brueckner2019}
G.~Br{\"{u}}ckner, M.~Himmel, and I.~Rutter.
\newblock An {SPQR}-tree-like embedding representation for upward planarity.
\newblock In D.~Archambault and C.~D. T{\'{o}}th, editors, {\em Proceedings of
  the 27th International Symposium on Graph Drawing and Network Visualization
  (GD'19)}, volume 11904 of {\em LNCS}, pages 517--531. Springer, 2019.
\newblock \href {https://doi.org/10.1007/978-3-030-35802-0_39}
  {\path{doi:10.1007/978-3-030-35802-0_39}}.

\bibitem{Chen1999}
Z.-Z. Chen, X.~He, and C.-H. Huang.
\newblock Finding double euler trails of planar graphs in linear time [{CMOS}
  {VLSI} circuit design].
\newblock In {\em Proceedings of the 40th Annual Symposium on Foundations of
  Computer Science (FOCS'99)}. IEEE, 1999.
\newblock \href {https://doi.org/10.1109/sffcs.1999.814603}
  {\path{doi:10.1109/sffcs.1999.814603}}.

\bibitem{Cortese2008}
P.~F. Cortese, G.~{Di Battista}, F.~Frati, M.~Patrignani, and M.~Pizzonia.
\newblock C-planarity of c-connected clustered graphs.
\newblock {\em Journal of Graph Algorithms and Applications}, 12(2):225--262,
  2008.
\newblock \href {https://doi.org/10.7155/jgaa.00165}
  {\path{doi:10.7155/jgaa.00165}}.

\bibitem{Battista1989}
G.~{Di Battista} and R.~Tamassia.
\newblock Incremental planarity testing.
\newblock In {\em Proceedings of the 30th Annual Symposium on Foundations of
  Computer Science (FOCS'89)}, pages 436 -- 441. {IEEE}, 1989.
\newblock \href {https://doi.org/10.1109/sfcs.1989.63515}
  {\path{doi:10.1109/sfcs.1989.63515}}.

\bibitem{Battista1990}
G.~{Di Battista} and R.~Tamassia.
\newblock On-line graph algorithms with {SPQR}-trees.
\newblock In {\em Proceedings of the 17th International Colloquium on Automata,
  Languages, and Programming (ICALP'90)}, pages 598--611. Springer, 1990.
\newblock \href {https://doi.org/10.1007/bfb0032061}
  {\path{doi:10.1007/bfb0032061}}.

\bibitem{Battista1996}
G.~{Di Battista} and R.~Tamassia.
\newblock On-line maintenance of triconnected components with {SPQR}-trees.
\newblock {\em Algorithmica}, 15(4):302--318, 1996.
\newblock \href {https://doi.org/10.1007/bf01961541}
  {\path{doi:10.1007/bf01961541}}.

\bibitem{Battista1996p}
G.~{Di Battista} and R.~Tamassia.
\newblock On-line planarity testing.
\newblock {\em {SIAM} Journal on Computing}, 25(5):956--997, 1996.
\newblock \href {https://doi.org/10.1137/s0097539794280736}
  {\path{doi:10.1137/s0097539794280736}}.

\bibitem{Didimo2020}
W.~Didimo, G.~Liotta, G.~Ortali, and M.~Patrignani.
\newblock Optimal orthogonal drawings of planar 3-graphs in linear time.
\newblock In {\em Proceedings of the 14th Annual {ACM}-{SIAM} Symposium on
  Discrete Algorithms (SODA'20)}, pages 806--825. SIAM, 2020.
\newblock \href {https://doi.org/10.1137/1.9781611975994.49}
  {\path{doi:10.1137/1.9781611975994.49}}.

\bibitem{Eppstein1996}
D.~Eppstein, Z.~Galil, G.~F. Italiano, and T.~H. Spencer.
\newblock Separator based sparsification.
\newblock {\em Journal of Computer and System Sciences}, 52(1):3--27, 1996.
\newblock \href {https://doi.org/10.1006/jcss.1996.0002}
  {\path{doi:10.1006/jcss.1996.0002}}.

\bibitem{Fedarko2017}
M.~Fedarko, J.~Ghurye, T.~Treagen, and M.~Pop.
\newblock Metagenomescope: Web-based hierarchical visualization of metagenome
  assembly graphs.
\newblock In F.~Frati and K.-L. Ma, editors, {\em Proceedings of the 25th
  International Symposium on Graph Drawing and Network Visualization (GD'17)},
  pages 630--632. Springer, 2017.
\newblock (Poster).
\newblock URL:
  \url{https://gd2017.ccis.northeastern.edu/files/posters/fedarko-metagenomescope.pdf},
  \href {https://doi.org/10.1007/978-3-319-73915-1}
  {\path{doi:10.1007/978-3-319-73915-1}}.

\bibitem{Feng1995}
Q.-W. Feng, R.~F. Cohen, and P.~Eades.
\newblock Planarity for clustered graphs.
\newblock In P.~G. Spirakis, editor, {\em Proceedings of the 3rd Annual
  European Symposium on Algorithms (ESA'95)}, volume 979 of {\em LNCS}, pages
  213--226. Springer, 1995.
\newblock \href {https://doi.org/10.1007/3-540-60313-1_145}
  {\path{doi:10.1007/3-540-60313-1_145}}.

\bibitem{Franken2005}
D.~Franken, J.~Ochs, and K.~Ochs.
\newblock Generation of wave digital structures for networks containing
  multiport elements.
\newblock {\em {IEEE} Transactions on Circuits and Systems I: Regular Papers},
  52(3):586--596, 2005.
\newblock \href {https://doi.org/10.1109/tcsi.2004.843056}
  {\path{doi:10.1109/tcsi.2004.843056}}.

\bibitem{Fulek2015}
R.~Fulek, J.~Kyn{\v{c}}l, I.~Malinovi{\'{c}}, and D.~P{\'{a}}lvölgyi.
\newblock Clustered planarity testing revisited.
\newblock {\em The Electronic Journal of Combinatorics}, 22(4), 2015.
\newblock \href {https://doi.org/10.37236/5002} {\path{doi:10.37236/5002}}.

\bibitem{Fulek2019}
R.~Fulek and C.~D. T{\'{o}}th.
\newblock Atomic embeddability, clustered planarity, and thickenability.
\newblock {\em Journal of the {ACM}}, 69(2):13:1--13:34, 2022.
\newblock \href {http://arxiv.org/abs/1907.13086v1}
  {\path{arXiv:1907.13086v1}}, \href {https://doi.org/10.1145/3502264}
  {\path{doi:10.1145/3502264}}.

\bibitem{Galil1999}
Z.~Galil, G.~F. Italiano, and N.~Sarnak.
\newblock Fully dynamic planarity testing with applications.
\newblock {\em Journal of the {ACM}}, 46(1):28--91, 1999.
\newblock \href {https://doi.org/10.1145/300515.300517}
  {\path{doi:10.1145/300515.300517}}.

\bibitem{Gutwenger2010}
C.~Gutwenger.
\newblock {\em Application of {SPQR}-trees in the planarization approach for
  drawing graphs.}
\newblock PhD thesis, 2010.
\newblock URL:
  \url{https://eldorado.tu-dortmund.de/bitstream/2003/27430/1/diss\_gutwenger.pdf}.

\bibitem{Gutwenger2002}
C.~Gutwenger, M.~J{\"{u}}nger, S.~Leipert, P.~Mutzel, M.~Percan, and
  R.~Weiskircher.
\newblock Advances in c-planarity testing of clustered graphs.
\newblock In S.~G. Kobourov and M.~T. Goodrich, editors, {\em Proceedings of
  the 10th International Symposium on Graph Drawing (GD'02)}, volume 2528 of
  {\em LNCS}, pages 220--235. Springer, 2002.
\newblock \href {https://doi.org/10.1007/3-540-36151-0_21}
  {\path{doi:10.1007/3-540-36151-0_21}}.

\bibitem{Gutwenger2001}
C.~Gutwenger and P.~Mutzel.
\newblock A linear time implementation of {SPQR}-trees.
\newblock In {\em Proceedings of the 8th International Symposium on Graph
  Drawing (GD'20)}, pages 77--90. Springer, 2001.
\newblock \href {https://doi.org/10.1007/3-540-44541-2_8}
  {\path{doi:10.1007/3-540-44541-2_8}}.

\bibitem{Holm2020p}
J.~Holm and E.~Rotenberg.
\newblock Fully-dynamic planarity testing in polylogarithmic time.
\newblock In K.~Makarychev, Y.~Makarychev, M.~Tulsiani, G.~Kamath, and
  J.~Chuzhoy, editors, {\em Proceedings of the 52nd Annual {ACM} {SIGACT}
  Symposium on Theory of Computing (STOC'20)}, volume abs/1911.03449, pages
  167--180. {ACM}, 2020.
\newblock \href {http://arxiv.org/abs/1911.03449} {\path{arXiv:1911.03449}},
  \href {https://doi.org/10.1145/3357713.3384249}
  {\path{doi:10.1145/3357713.3384249}}.

\bibitem{Holm2020}
J.~Holm and E.~Rotenberg.
\newblock Worst-case polylog incremental {SPQR}-trees: Embeddings, planarity,
  and triconnectivity.
\newblock In {\em Proceedings of the 14th Annual {ACM}-{SIAM} Symposium on
  Discrete Algorithms (SODA'20)}, pages 2378--2397. SIAM, 2020.
\newblock \href {https://doi.org/10.1137/1.9781611975994.146}
  {\path{doi:10.1137/1.9781611975994.146}}.

\bibitem{Hopcroft1973}
J.~E. Hopcroft and R.~E. Tarjan.
\newblock Dividing a graph into triconnected components.
\newblock {\em {SIAM} Journal on Computing}, 2(3):135--158, 1973.
\newblock \href {https://doi.org/10.1137/0202012} {\path{doi:10.1137/0202012}}.

\bibitem{Lengauer1989}
T.~Lengauer.
\newblock Hierarchical planarity testing algorithms.
\newblock {\em Journal of the ACM}, 36(3):474--509, 1989.
\newblock \href {https://doi.org/10.1145/65950.65952}
  {\path{doi:10.1145/65950.65952}}.

\bibitem{Liotta2020}
G.~Liotta, I.~Rutter, and A.~Tappini.
\newblock Simultaneous {FPQ}-ordering and hybrid planarity testing.
\newblock In {\em Proceedings of the 46th International Conference on Current
  Trends in Theory and Practice of Informatics (SOFSEM'20)}, pages 617--626.
  Springer, 2020.
\newblock \href {https://doi.org/10.1007/978-3-030-38919-2_51}
  {\path{doi:10.1007/978-3-030-38919-2_51}}.

\bibitem{MacLane1937}
S.~Mac~Lane.
\newblock A structural characterization of planar combinatorial graphs.
\newblock {\em Duke Mathematical Journal}, 3(3):460--472, 1937.
\newblock \href {https://doi.org/10.1215/S0012-7094-37-00336-3}
  {\path{doi:10.1215/S0012-7094-37-00336-3}}.

\bibitem{Mutzel2003}
P.~Mutzel.
\newblock The {SPQR}-tree data structure in graph drawing.
\newblock In J.~C.~M. Baeten, J.~K. Lenstra, J.~Parrow, and G.~J. Woeginger,
  editors, {\em Proceedings of the 30th International Colloquium on Automata,
  Languages and Programming (ICALP'03)}, volume 2719 of {\em LNCS}, pages
  34--46. Springer, 2003.
\newblock \href {https://doi.org/10.1007/3-540-45061-0_4}
  {\path{doi:10.1007/3-540-45061-0_4}}.

\bibitem{Poutre1992}
J.~A.~L. Poutr{\'{e}}.
\newblock Maintenance of triconnected components of graphs.
\newblock In {\em Proceedings of the 19th International Colloquium on Automata,
  Languages and Programming (ICALP'92)}, pages 354--365. Springer, 1992.
\newblock \href {https://doi.org/10.1007/3-540-55719-9_87}
  {\path{doi:10.1007/3-540-55719-9_87}}.

\bibitem{Poutre1994}
J.~A.~L. Poutr{\'{e}}.
\newblock Alpha-algorithms for incremental planarity testing (preliminary
  version).
\newblock In {\em Proceedings of the 26th annual {ACM} symposium on Theory of
  computing (STOC'94)}. {ACM} Press, 1994.
\newblock \href {https://doi.org/10.1145/195058.195439}
  {\path{doi:10.1145/195058.195439}}.

\bibitem{Vanhatalo2009}
J.~Vanhatalo, H.~V{\"{o}}lzer, and J.~Koehler.
\newblock The refined process structure tree.
\newblock {\em Data and Knowledge Engineering}, 68(9):793--818, 2009.
\newblock \href {https://doi.org/10.1016/j.datak.2009.02.015}
  {\path{doi:10.1016/j.datak.2009.02.015}}.

\bibitem{Manteuffel2012}
A.~von Manteuffel and C.~Studerus.
\newblock Reduze 2 - distributed feynman integral reduction.
\newblock 2012.
\newblock \href {http://arxiv.org/abs/1201.4330} {\path{arXiv:1201.4330}}.

\bibitem{Weiskircher2002}
R.~Weiskircher.
\newblock {\em New applications of {SPQR}-trees in graph drawing}.
\newblock PhD thesis, Universität des Saarlandes, 2002.
\newblock \href {https://doi.org/10.22028/D291-25752}
  {\path{doi:10.22028/D291-25752}}.

\bibitem{Westbrook1992}
J.~Westbrook.
\newblock Fast incremental planarity testing.
\newblock In {\em Proceedings of the 19th International Colloquium on Automata,
  Languages and Programming (ICALP'92)}, pages 342--353. Springer, 1992.
\newblock \href {https://doi.org/10.1007/3-540-55719-9_86}
  {\path{doi:10.1007/3-540-55719-9_86}}.

\bibitem{Zhang2013}
Y.~Zhang, W.~Luk, H.~Zhou, C.~Yan, and X.~Zeng.
\newblock Layout decomposition with pairwise coloring for multiple patterning
  lithography.
\newblock In J.~Henkel, editor, {\em Proceedings of the {IEEE/ACM}
  International Conference on Computer-Aided Design (ICCAD'13)}, pages
  170--177. {IEEE}, 2013.
\newblock \href {https://doi.org/10.1109/ICCAD.2013.6691115}
  {\path{doi:10.1109/ICCAD.2013.6691115}}.

\end{thebibliography}

\ifthenelse{\boolean{long}}{}{
\newpage
\appendix
\input{appendix}

\section{Background and Discussion}
In this section, we will give some background on the historical development of and further details on the problems \cplan and \pqplan together with summary of the algorithm of Bläsius et al. for solving both problems.
Furthermore, we will discuss how our and also previous work on dynamic SPQR-trees can be used in the context of both problems.

}

\end{document}